\documentclass[11pt, draftcls, one column]{IEEEtran}

\usepackage{dsfont}
\usepackage{authblk}
\usepackage{subfig}
\def\withnotes{1}

\usepackage{makecell}
\usepackage{bbm}
\usepackage{graphicx}
\usepackage[usenames,dvipsnames]{color}
\usepackage{epsfig}
\usepackage{amssymb}
\usepackage{amsmath}
\usepackage{amsthm}
\usepackage{latexsym}
\usepackage{setspace}
\usepackage{bbm}
\usepackage{flushend}
\usepackage{pgfplots}
\usepackage[top=1in, bottom=1in, left=1in, right=1in]{geometry}
\usepackage{float}
\usepackage{mathpazo}

\newcommand{\bPr}[1]{{\mathrm{Pr}}\left(#1\right)}

\newcommand{\cE}{{\mathcal E}}

\newcommand{\mN}{{\mathbbm N}}
\newcommand{\mR}{{\mathbbm R}}
\newcommand{\cP}{{\mathcal P}}

\newcommand{\cX}{{\mathcal X}}

\newcommand{\bx}{\mathbf{x}}
\newcommand{\by}{\mathbf{y}}

\newcommand{\ep}{\epsilon}

\makeatletter
\newtheorem*{rep@theorem}{\rep@title}
\newcommand{\newreptheorem}[2]{%
\newenvironment{rep#1}[1]{%
 \def\rep@title{#2 \ref{##1}}%
 \begin{rep@theorem}}%
 {\end{rep@theorem}}}
\makeatother

\newtheorem{theorem}{Theorem}

\newtheorem*{corollary*}{Corollary}

\newtheorem{lemma}[theorem]{Lemma}
\newtheorem*{lemma*}{Lemma}

\newreptheorem{theorem}{Theorem}
\newreptheorem{lemma}{Lemma}

\theoremstyle{remark}
\newtheorem*{remark*}{Remark}
\newtheorem*{remarks*}{Remarks}

\newtheorem{example}{Example}
\theoremstyle{definition}

\newcommand{\ed}{\stackrel{{\rm def}}{=}}

\def\undertilde#1{\mathord{\vtop{\ialign{##\crcr
$\hfil\displaystyle{#1}\hfil$\crcr\noalign{\kern1.5pt\nointerlineskip}
$\hfil\tilde{}\hfil$\crcr\noalign{\kern1.5pt}}}}}

\definecolor{comments}{rgb}{1,0.5,0}

\usepackage[colorinlistoftodos,textsize=scriptsize]{todonotes}
\ifnum\withnotes=1
\else

\fi

\newcommand{\ignore}[1]{}

\newcommand{\stack}[2]{\stackrel{#1}{\smash{#2}\rule{0pt}{1.1ex}}}
\newcommand{\dblstk}[2]{\stackrel{#1}{\smash{\stack{#1}{#2}}\rule{0pt}{1.5ex}}}
\newcommand{\dbltldOmg}{\dblstk \sim\Omega}

\newcommand{\reals}{\mathbb{R}}
\newcommand{\naturals}{\mathbb{N}}
\newcommand{\Paren}[1]{\left(#1\right)}
\newcommand{\Sparen}[1]{\left[#1\right]}
\newcommand{\Sets}[1]{\left\{#1\right\}}
\newcommand{\Order}{O}
\newcommand{\order}{o}

\newcommand{\Const}{C}

\newcommand{\integers}{\mathbb{Z}}

\newcommand{\indicator}{\mathds 1} % or use \mathbbm1
\newcommand{\upto}{,\ldots,}
\newcommand{\thr}{\tau}
\newcommand{\polapxerr}{E_\degree}

\newcommand{\powsumsmb}{P}
\newcommand{\estimator}{f}
\newcommand{\estmmnemp}{\widehat \powsumsmb_\renprm^\text{e}}
\newcommand{\estmmnplndt}{\widehat \powsumsmb_\renprm^{d,\thr}}
\newcommand{\estmmnunb}{\widehat \powsumsmb_\renprm^\text{u}}

\newcommand{\estrenemp}{\Fnc_\renprm^\text{e}}
\newcommand{\estrenplndt}{\Fnc_\renprm^{d,\thr}}
\newcommand{\estrenunb}{\Fnc_\renprm^\text{u}}

\newcommand{\sample}{S}
\newcommand{\safkde}{\sample_\alpha^f(k,\delta,\epsilon)}
\newcommand{\sakde}{\sample_\alpha(k,\delta,\epsilon)}
\newcommand{\sak}{\sample_\alpha(k)}
\newcommand{\smak}{\sample^{\powsumsmb\times}_\alpha(k)}
\newcommand{\saak}{\sample^{\powsumsmb+}_\alpha(k)}

\newcommand{\ent}{H}
\newcommand{\rent}[2]{\ent_{#1}(#2)}
\newcommand{\rental}[1]{\rent{\alpha}{#1}}
\newcommand{\renprm}{\alpha}

\newcommand{\renyi}{R{\'e}nyi entropy }

\newcommand{\absz}{k}
\newcommand{\nsmp}{n}
\newcommand{\Nsmp}{N}

\newcommand{\Pfl}{\Phi}
\newcommand{\Pfll}{\Pfl_l}

\newcommand{\dm}{d}
\newcommand{\smb}{x}
\newcommand{\Mlt}{N}
\newcommand{\Mlts}[1]{\Mlt_{#1}}
\newcommand{\Mltsmb}{\Mlts{\smb}}

\newcommand{\Mltx}{\Mlt_x}

\newcommand{\Xon}{X^n}

\newcommand{\XoN}{X^N}

\renewcommand{\bPr}[1]{{\PP}\left(#1\right)}
\newcommand{\dP}{\mathrm p}
\newcommand{\dQ}{\mathrm q}

\newcommand{\dPs}[1]{\dP_{#1}}

\newcommand{\dPx}{\dP_{x}}
\newcommand{\dQx}{\dQ_{x}}
\newcommand{\hdPx}{\hat{\dP}_{x}}

\newcommand{\zipf}{Z}
\newcommand{\zipfbk}{\zipf_{\beta, \absz}}
\newcommand{\zipfsk}[1]{\zipf_{ #1, \absz}}

\newcommand{\poisson}{{\rm Poi}}
\newcommand{\psnlmb}{\poisson(\lambda)}
\newcommand{\psnlmbs}[1]{\poisson(\lambda, #1)}
\newcommand{\psns}[1]{\poisson(#1)}

\newcommand{\poid}[1]{\poisson(#1)}

\newcommand{\poiprob}[2]{e^{-#1}\cdot\frac{#1^{#2}}{#2!}}
\newcommand{\npsmb}{\lambda_{\smb}}
\newcommand{\norm}{{\powsumsmb}}

\newcommand{\normP}[1]{\norm_{#1}(\dP)}
\newcommand{\normQ}[1]{\norm_{#1}(\dQ)}

\newcommand{\norma}{\norm_{a}}
\newcommand{\mntests}[1]{\widehat{\rm \powsumsmb}_{#1}}
\newcommand{\mmntest}{\widehat{\rm \powsumsmb}_{\renprm}}
\newcommand{\mntestmdn}{\widehat{\rm \powsumsmb}^{\prime}_{\renprm}}

\newcommand{\Fnc}{f}

\newcommand{\hfnc}{\hat{\Fnc}}
\newcommand{\tfnc}{\tilde{\Fnc}}

\newcommand{\Gnc}{g}

\newcommand{\fncalXn}[1]{f_{\renprm}(\Xon)}

\newcommand{\EE}{\mathbb{E}}
\newcommand{\expectation}[1]{\EE\!\Sparen{#1\,}}
\newcommand{\PP}{\mathbb{P}}
\newcommand{\variance}[1]{\Var\!\Sparen{#1\,}}
\newcommand{\Var}{\mathrm{\mathbb{V}ar}}

\newcommand{\esterr}{\delta}
\newcommand{\prerr}{\epsilon}

\newcommand{\flnpwrprm}{r}
\newcommand{\flnpwrss}[2]{#1^{\underline{#2}}}
\newcommand{\fpXf}{\flnpwrss X\flnpwrprm}

\newcommand{\flnpwrrp}[1]{\flnpwrss{#1}{\renprm}}

\newcommand{\sam}{S}

\newcommand{\degree}{d}
\newcommand{\poly}{q}
\newcommand{\polyx}{\poly(x)}
\newcommand{\polys}[1]{\poly(#1)}
\newcommand{\coefm}{a_m}

\IEEEoverridecommandlockouts
\date{}

\definecolor{myred}{HTML}{FF0000}
\definecolor{mygreen}{HTML}{008000}
\begin{document}
%\addtocounter{page}{-1}

{\renewcommand{\thefootnote}{}\footnotetext{
An initial version of this paper \cite{AcharyaOTT15} was presented at
the ACM Symposium on Discrete Algorithms (SODA), 2015. 
}}

% This is the title page and contains the title, followed by names and
% emails of the authors and some sort of an abstract

%\title{The Complexity of Estimating R{\'e}nyi Entropy}
\title{Estimating Renyi Entropy of Discrete Distributions}

\author[1]{Jayadev Acharya}
\author[2]{Alon Orlitsky}
\author[2]{Ananda Theertha Suresh}
\author[3]{Himanshu Tyagi}
\affil[1]{Massachusetts Institute of Technology (jayadev@csail.mit.edu)}
\affil[2]{University of California, San Diego (\{alon, asuresh\}@ucsd.edu)}
\affil[3]{Indian Institute of Science (htyagi@ece.iisc.ernet.in)}

\maketitle

\begin{abstract}
It was recently shown that estimating the Shannon entropy $H(\dP)$ of a
discrete $k$-symbol distribution $\dP$ requires $\Theta(k/\log k)$ samples,
a number that grows near-linearly in the support size.
In many applications $H(\dP)$ can be replaced by
the more general R\'enyi entropy of order $\renprm$, $\rental \dP$.
We determine the number of samples needed to estimate $\rental\dP$
for all $\renprm$, showing that $\renprm < 1$ 
requires a super-linear, roughly $k^{1/\renprm}$ samples,
noninteger $\renprm>1$ requires a near-linear $k$ samples,
but, perhaps surprisingly, integer $\alpha>1$ requires only
$\Theta(k^{1-1/\alpha})$ samples.
Furthermore, developing on a recently established connection between polynomial approximation
and estimation of additive functions of the form $\sum_xf (\dPx)$, we reduce the sample complexity
for noninteger values of $\alpha$ by a factor of $\log k$ compared to
the empirical estimator. 
The estimators achieving these
bounds are simple and run in time linear in the number of samples. 
Our lower bounds provide explicit
  constructions of distributions with different R\'enyi entropies that are
 hard to distinguish.
\end{abstract}

\section{Introduction}\label{s:introduction}
\subsection{Shannon and R\'enyi entropies}
One of the most commonly used measure of randomness of a distribution $\dP$ over
a discrete set $\cX$ is its \emph{Shannon entropy}
\[
H(\dP)
\ed
\sum_{\smb\in\cX} \dPx\log \frac1\dPx.
\]
The estimation of Shannon entropy has several applications,
including measuring genetic diversity~\cite{ShenkinEM91},
quantifying neural activity~\cite{Paninski03, NemenmanBRS04}, 
network anomaly detection~\cite{LallSOXZ06}, and others.
It was recently shown that estimating the Shannon entropy 
of a discrete distribution $\dP$ over $\absz$ elements to a given additive accuracy
requires
%\footnote{$f(k)=\Theta(g(k))$ if there  exist constants $c$ and $C$ such that 
%$cg(k)\le f(k)\le Cg(k)$.}
$\Theta(\absz/\log \absz)$ independent samples from
$\dP$~\cite{Paninski04b,Valiant11b};
see~\cite{JiaoVW14,WuY14} for subsequent extensions.
This number of samples grows near-linearly with the alphabet size and
is only a logarithmic factor smaller than the $\Theta(\absz)$ samples needed to learn $\dP$
itself to within a small statistical distance.

A popular generalization of Shannon entropy is the \emph{R\'enyi
  entropy} of order $\renprm\ge0$, defined for $\renprm\ne1$ by
\[
\rental \dP
\ed
\frac{1}{1-\renprm} \log \sum_{\smb\in\cX} \dPx^\renprm
\]
and for $\renprm=1$ by 
\[
H_1( \dP)
\ed
\lim_{\renprm\to1} \rental\dP.
\]
It was shown in the seminal paper~\cite{Ren61}
that \renyi of order 1
is Shannon entropy, namely $H_1(\dP)=H(\dP)$, and
for all other orders it is the unique extension of Shannon entropy
when of the four requirements in Shannon entropy's axiomatic definition,
continuity, symmetry, and normalization are kept
but grouping is restricted to only additivity over independent random variables ($cf.$ \cite{IlicS13}). 

\renyi too has many applications.
It is often used as a bound on Shannon entropy~\cite{Mokkadem89,
NemenmanBRS04, HarveyNO08}, and in many applications it replaces Shannon
entropy as a measure of randomness~\cite{Csiszar95, Massey94, Arikan96}. 
It is also of interest in its own right, with diverse applications 
to unsupervised learning~\cite{Xu98, JenssenHEPE03},
source adaptation~\cite{MansourMR12},
image registration~\cite{MaHGM00, NeemuchwalaHZC06},
and password guessability~\cite{Arikan96, PfisterS04, HanawalS11}
among others.
In particular, the \renyi of order 2, $H_2(\dP)$,
measures the quality of random number
generators~\cite{Knuth73, OorschotW99}, 
determines the number of unbiased bits that can be extracted
from a physical source of randomness~\cite{ImpagliazzoZ89,BennettBCM95},
helps test graph expansion~\cite{GoldreichR00} and
closeness of distributions~\cite{BatuFRSW13, Paninski08},
and characterizes the number of reads needed to reconstruct a DNA
sequence~\cite{MotahariBT13}. 

Motivated by these and other applications, 
unbiased and heuristic estimators of \renyi
have been studied in the physics literature 
following~\cite{Grassberger88}, and 
asymptotically consistent and normal
estimates were
proposed in~\cite{XuE10,KallbergLS11}.
However, no systematic study of the complexity
of estimating R\'enyi entropy is available. 
For example, it was hitherto unknown if the number of samples needed
to estimate the R\'enyi entropy of a given order $\renprm$ differs
from that required for Shannon 
entropy, or whether it varies with the order $\renprm$, 
or how it depends on the alphabet size $\absz$.

\subsection{Definitions and results}
\label{sct:dfn_rsl}
We answer these questions by showing that the number of samples needed
to estimate 
%the \renyi
$\rental\dP$ falls into three different ranges.
For $\renprm<1$ it grows super-linearly with $\absz$,
for $1<\renprm\not\in\integers$ it grows almost linearly with $\absz$,
and most interestingly, for the popular orders
$1<\renprm\in\integers$ it grows as
$\Theta(\absz^{1-1/\renprm})$,
which is much less than the sample complexity of estimating 
Shannon entropy.

To state the results more precisely we need a few definitions.
A R\'enyi-entropy \emph{estimator} for distributions over support
set $\cX$ is a function $\estimator:\cX^*\to\reals$ mapping a 
sequence of samples drawn from a distribution to an estimate of its
entropy. 
The sample complexity of an estimator $\estimator$ 
for distributions over $\absz$ elements is defined as 
\[
\safkde
\ed
\min_n
\Sets{n:
\dP\Paren{|\rental{\dP} - \estimator\left(\Xon\right)| > \esterr}
<
\prerr, \forall\, \dP \text{ with }\|p\|_0\le k},
\]
$i.e.$, the minimum number of samples required by $\estimator$ to estimate
$\rental \dP$ of any $\absz$-symbol distribution
$\dP$ to  a given additive accuracy $\esterr$ with probability greater
than $1 - \prerr$.
The \emph{sample complexity} of estimating $\rental\dP$ is then 
\[
\sakde
\ed
\min_f
\safkde,
\]
the least number of samples any estimator needs to estimate 
$\rental \dP$ for 
all $\absz$-symbol distributions $\dP$,
to an additive accuracy $\esterr$ and with probability greater than $1 -
\prerr$. This is a min-max definition where the goal is
 to obtain the \emph{best} estimator for the \emph{worst}
  distribution.

The desired accuracy $\esterr$ and confidence
 $1-\prerr$ are typically fixed. We are therefore most
interested\footnote{Whenever a more refined result indicating the
  dependence of sample complexity on both $\absz$ and $\esterr$ is
  available, we shall use the more elaborate $\sakde$ notation.} in
the dependence of $\sakde$ on the alphabet   
size $\absz$ and omit the dependence of $\sakde$ on $\esterr$ and
$\prerr$ to write $\sak$. 
In particular, we are interested in the {\it large alphabet} regime 
and focus on the essential growth rate of $\sak$ as a function
  of $\absz$ for large $\absz$.
Using the standard asymptotic notations, let
$\sak=\Order(\absz^\beta)$ indicate that for some constant $c$
which may depend on $\renprm$, $\esterr$, and $\prerr$, 
for all sufficiently large $\absz$, $\sakde \le c\cdot \absz^\beta$.
Similarly, $\sak=\Theta(\absz^\beta)$ adds the corresponding $\Omega(\absz^\beta)$ lower bound
for $\sakde$, for all sufficiently small $\esterr$ and $\prerr$.
Finally, extending the $\tilde\Omega$ 
notation\footnote{The notations $\tilde
  O$, $\tilde \Omega$, and $\tilde \Theta$ hide poly-logarithmic factors.}, we let
$\sak=\,\dbltldOmg(\absz^\beta)$ indicate that for every sufficiently small $\prerr$ and arbitrary
$\eta>0$, there exist $c$ and $\esterr$ depending on $\eta$
such that for all $k$ sufficiently large 
$\sakde>c\absz^{\beta-\eta}$, namely 
$\sak$ grows polynomially in $\absz$ with exponent not less than $\beta-\eta$
for $\esterr\le \delta_\eta$.

We show that $\sak$ behaves differently in three ranges of $\renprm$.
For $0\le\renprm<1$, 
\[
\dbltldOmg
\Paren{\absz^{1/\renprm}}
\le
\sak
\le
\Order\Paren{\frac{\absz^{1/\renprm}}{\log \absz}},
\]
namely the sample complexity grows super-linearly in $\absz$ and
estimating the R\'enyi entropy of these orders is even more
difficult than estimating the
Shannon entropy. 
In fact, the upper bound follows from a corresponding result on
estimation of power sums considered in~\cite{JiaoVW14}
(see Section~\ref{sec:nonintegral} for further discussion). 
For completeness, we show in  Theorem~\ref{t:upper_bounds_alpha_small}
that the empirical estimator requires $O(\absz^{1/\renprm})$ samples and 
in Theorem~\ref{t:upper_bounds_arbitrary_poly_small} prove the
improvement by a factor of $\log \absz$.
%% with smaller constants than
%% implied by~\cite{JiaoVW14}. 
The lower bound is proved in
Theorem~\ref{t:lower_bounds_alpha_small}. 
%\jcolor{
%We first show in
%  Theorem~\ref{t:upper_bounds_alpha_small} that the empirical
%  estimator requires $O(\absz^{1/\renprm})$ samples, and subsequently
%in Theorem~\ref{t:upper_bounds_arbitrary_poly_small} prove the
%improvement by a factor of $\log \absz$ using polynomial
%approximation estimators. 
%
%
%The upper bound for $\renprm<1$ can be derived
%from~\cite{JiaoVW14}(See Section~\ref{sct:rlt_mmn_est} for further
%discussion), 
%and a subsequent result that the simple empirical-frequency estimator
%requires $\Order\Paren{\absz^{1/\renprm}}$ samples is shown in
%Theorem~\ref{t:upper_bounds_alpha_small}. 
%The lower bound is proved in
%Theorem~\ref{t:lower_bounds_alpha_small}.}

For $1<\renprm\notin\naturals$, 
\[
\dbltldOmg(\absz)
\le
\sak
\le
\Order\Paren{\frac{\absz}{\log \absz}},
\]
namely as with Shannon entropy, the sample complexity 
grows roughly linearly in the alphabet size.
The lower bound is proved in Theorem~\ref{t:lower_bounds_arbitrary}. 
In a conference version of this paper \cite{AcharyaOTT15}, a weaker
$\Order(\absz)$ upper bound was established using the
empirical-frequency estimator. 
For the sake of completeness, we include this result as
Theorem~\ref{t:upper_bounds_arbitrary}. The tighter upper bound
reported here  
uses the best polynomial approximation based estimator of
\cite{JiaoVW14, WuY14} and is proved in
Theorem~\ref{t:upper_bounds_arbitrary_poly}. In fact, in the
Appendix we show that the empirical  estimator can't attain this $\log
\absz$ improvement and requires $\Omega(\absz/\esterr)$ and $\Omega((\absz/\esterr)^{1/\renprm})$
samples for $\renprm>1$ and $\renprm<1$, respectively. 

For $1 < \renprm \in\naturals$ and $\esterr$ and $\prerr$ sufficiently
small,
\[
\sam_\renprm(\absz, \esterr, \prerr)
=
\Theta\Paren{\frac{\absz^{1 - 1/\renprm}}{\esterr^2}},
\]
and in particular, the sample complexity is \emph{strictly sublinear} in the alphabet size.
The upper and lower bounds are shown in Theorems~\ref{t:upper_bounds_integer}
and~\ref{t:lower_bounds_integer2}, respectively. Figure~\ref{f:exponent} illustrates 
our results for different ranges of $\renprm$.

\pgfplotsset{compat=newest}

\begin{figure}[h]
\begin{center}
\begin{tikzpicture}[scale=1.0]\label{f:exponent}
\begin{axis}[width=13cm,height=7cm,xmin=0, xmax=7,ymax=4,ymin=0, samples=60,
  xlabel=$\alpha\rightarrow$, ylabel=$\frac{\log (\sak)}{\log (\absz)}\rightarrow$]
  \addplot[red, dashed][domain=0.3:0.4] (x,1/x);
  \addplot[red, thick][domain=0.4:1] (x,1/x);
  \addplot[blue, thick] [domain=1:7](x,1);
  \addplot[blue, mark=*, mark options={fill=white}] coordinates
  {(2,1)};
  \addplot[blue, mark=*, mark options={fill=white}] coordinates
  {(3,1)};
  \addplot[blue, mark=*, mark options={fill=white}] coordinates
  {(4,1)};
  \addplot[blue, mark=*, mark options={fill=white}] coordinates
  {(5,1)};
  \addplot[blue, mark=*, mark options={fill=white}] coordinates
  {(6,1)};
  \addplot[color=comments, mark=*] coordinates
  {(2,.5)};
  \addplot[color=comments, mark=*] coordinates
  {(3,0.6666)};
  \addplot[color=comments, mark=*] coordinates
  {(4,0.75)};
  \addplot[color=comments, mark=*] coordinates
  {(5,0.80)};
  \addplot[color=comments, mark=*] coordinates
  {(6,0.83333)};
\end{axis}
\end{tikzpicture}
\end{center}
\caption{Exponent of $\absz$ in $\sak$ as a function of $\renprm$.}
\end{figure}

Of the three ranges, the most frequently used, and coincidentally
the one for which the results are most surprising, is the last with
$\renprm=2,3,\ldots$. Some elaboration is in order.

First, for all integral $\alpha>1$, 
$\rental\dP$ can be estimated with a sublinear number of samples.
The most commonly used R\'enyi entropy, $H_2(\dP)$, can be estimated
within $\esterr$ using just
$\Theta\left(\sqrt \absz/\esterr^2\right)$ samples,
%In particular, the $\Order(\sqrt \absz)$ samples required by the estimator of \cite{GoldreichR00, BatuFRSW13} is the least required for estimating \renyi of any order.
%the least required for estimating \renyi of any order.
and hence R\'enyi entropy can be estimated much more efficiently than
Shannon Entropy, a useful property for large-alphabet applications
such as language  processing genetic analysis.

%% Second, when estimating Shannon entropy using $\Theta(\absz/\log \absz)$ samples,
%% the implicit constant factors are fairly
%% high (in the orders of $10^{6}$). For R\'enyi entropy of orders
%% $\renprm = 2, 3, ...$, the constants implied by  
%% $\Theta(\absz^{1 - 1/\renprm})$ are shown to be small in Theorem~\ref{t:upper_bounds_integer}.
%% Furthermore, the experiments
%% described below suggest that they may be even lower.

Also, note that R\'enyi entropy is continuous in the order $\renprm$. Yet
the sample complexity is discontinuous at integer orders. While this 
makes the estimation of the popular integer-order entropies easier,
it may seem contradictory. For instance, to approximate
$H_{2.001}( \dP)$ one
could approximate $H_2( \dP)$ using significantly fewer samples.
The reason for this is that the R\'enyi entropy, while continuous in
$\renprm$, is not uniformly
continuous. In fact, as shown in Example~\ref{exm:zpf_rny}, the difference between say
$H_2( \dP)$ and $H_{2.001}( \dP)$ may increase to infinity when the alphabet-size
increases.

It should also be noted that the estimators achieving the upper bounds
are simple and run in time linear in the number of samples. 
Furthermore, the estimators are {\it universal} in that they 
do not require the knowledge of $\absz$. On the other hand,
the lower bounds on $\sam_\renprm(\absz)$ hold even if
the estimator knows $\absz$.

\subsection{The estimators}\label{s:estimators}
The \emph{power sum} of order $\renprm$ of a distribution
$\dP$ over $\cX$ is
\[
\normP \renprm \ed \sum_{\smb\in\cX} \dPx^\renprm,
\]
and is related to the R\'enyi entropy for $\renprm\ne1$ via
\[
\rental\dP
=
\frac1{1-\renprm}
\log
\normP\renprm.
\]
Hence estimating $\rental\dP$ to an additive accuracy of $\pm\delta$ is
equivalent to estimating $\normP \renprm$ to a multiplicative accuracy
of $2^{\pm\delta\cdot(1-\renprm)}$. Furthermore, if $\delta(\alpha-1) \leq 1/2$ 
then estimating $\normP \renprm$ to multiplicative accuracy of
$1\pm\delta(1-\alpha)/2$ ensures a $\pm \delta$ additive accurate
estimate of $\rental \dP$. 

We construct estimators for the power-sums of distributions with
a multiplicative-accuracy of $(1\pm \esterr)$  and hence obtain
an additive-accuracy of $\Theta(\delta)$ for R\'enyi entropy estimation. 
We consider the following three different 
estimators for different ranges of $\renprm$ 
and with different performance  guarantees.

\paragraph{Empirical estimator} The \emph{empirical}, or \emph{plug-in}, estimator of $\normP\renprm$
is given by
\begin{align}
\label{eqn:empirical}
%\Fnc_\renprm(\Xon)
%=
%\frac{1}{1-\renprm} \log \sum_\smb \left(\frac{\Mltsmb}{\nsmp}\right)^\renprm.
\estmmnemp
\ed
\sum_\smb \left(\frac{\Mltsmb}{\nsmp}\right)^\renprm.
\end{align}
For $\alpha\ne 1$, $\estmmnemp$ is a not an unbiased estimator of
$\normP\renprm$. 
However, we show in Theorem~\ref{t:upper_bounds_alpha_small} that for
$\renprm<1$ the sample complexity of the empirical estimator is
$\Order(\absz^{1/\renprm})$ and in 
Theorem~\ref{t:upper_bounds_arbitrary} that for $\renprm>1$ it
 is $\Order(\absz)$. In the appendix, we show matching
lower bounds thereby characterizing the $\absz$-dependence of the sample complexity
of empirical estimator.

%% Using the lower bounds in Section~\ref{s:lower_bounds}, we prove that
%% the empirical estimator achieves the optimal exponent of $\absz$ for
%% all $\renprm\notin\mathbb{N}$. 

\paragraph{Bias-corrected estimator}
For integral $\renprm>1$, the \emph{bias-corrected} estimator for $\normP\renprm$ is
\begin{align}
\label{eqn:bias-corrected}
\estmmnunb
\ed
\sum_x \frac{\flnpwrss\Mltsmb\renprm}{n^\renprm},
\end{align}
where for integers $N$ and $r>0$, $\flnpwrss N r\ed N(N-1)\ldots
(N-r+1)$. A variation of this estimator was proposed first in \cite{Bar-YossefKS01} 
for estimating moments of frequencies in a sequence using random samples drawn from it.
Theorem~\ref{t:upper_bounds_integer} shows that for $1<\renprm\in\integers$,
$\estmmnunb$ estimates $\normP\renprm$ within a factor of $1\pm
\esterr$ using $\Order(\absz^{1-1/\renprm}/\esterr^2)$ samples, and 
Theorem~\ref{t:lower_bounds_integer2} shows that this number
is optimal up to a constant factor.

\paragraph{Polynomial approximation estimator}
To obtain a logarithmic improvement in $\sak$, we consider
the polynomial approximation estimator proposed in \cite{WuY14, JiaoVW14}
for different problems, concurrently to a conference version
\cite{AcharyaOTT15} of this paper. 
The polynomial approximation estimator first considers the \emph{best
  polynomial approximation} of degree $\degree$ to $y^\renprm$ for the interval
$y\in[0,1]$~\cite{Timan63}. Suppose this polynomial is given by
$a_0+a_1y+a_2y^2+\ldots+a_\degree y^\degree$.
We roughly divide the samples into two parts. Let $\Mltsmb^\prime$
and $\Mltsmb$ be the multiplicities of $\smb$ in the first and second
parts respectively. 
The polynomial approximation estimator uses the empirical estimate of
$\dPx^\renprm$ for large $\Mltsmb^\prime$, but estimates a polynomial
approximation of $\dPx^\renprm$ for a small $\Mltsmb^\prime$; the
integer powers of $\dPx$ in the latter in turn is estimated using the bias-corrected estimator.

The estimator is roughly of the form
\begin{align}
\label{eqn:polynomial}
\estmmnplndt \ed \sum_{\smb: \Mltsmb^\prime \le
  \thr}\left(\sum_{m=0}^\degree \frac{\coefm(2\thr)^{\renprm -m}
    \flnpwrss \Mltsmb m}{\nsmp^\renprm}\right)  
+ \sum_{\smb :\Mltsmb^\prime > \thr}
\left(\frac{\Mltsmb}{\nsmp}\right)^\renprm, 
\end{align}
where $\degree$ and $\tau$ are both $O(\log n)$ and chosen
appropriately. 

Theorem~\ref{t:upper_bounds_arbitrary_poly} and Theorem~\ref{t:upper_bounds_arbitrary_poly_small} show that for $\renprm>1$
and $\renprm<1$, respectively, 
the sample complexity of $\estmmnplndt$ is $\Order(\absz/\log \absz)$ and $\Order(\absz^{\frac 1 \renprm}/\log \absz)$,
resulting in a reduction in sample complexity of $\Order(\log \absz)$ over
the empirical estimator. 

Table~\ref{tab:summary} summarizes the performance of these
  estimators in terms of their sample complexity. The last column
  denote the lower bounds from Section~\ref{s:lower_bounds}.

Our goal in this work was to identify the exponent of $\absz$ in $\sak$. In
the process, we were able to characterize the sample complexity
$\sakde$ for $1<\renprm\in \mN$. However, we only
obtain partial results towards characterizing the sample complexity
$\sakde$ for a general $\renprm$. Specifically, while we show that the
empirical estimator attains the aforementioned exponent for every
$0<\renprm\notin \mN$, we note that the polynomial approximation
estimator has a lower sample complexity than the empirical
estimator. The exact characterization of $\sakde$ for a
general $\renprm$ remains open.

\renewcommand{\arraystretch}{1.6}

\begin{table}
\begin{center}
\begin{tabular}{|c|c|c|c|c|}
\hline
\textbf{Range of $\renprm$}&\textbf{Empirical} & \textbf{Bias-corrected} &
\textbf{Polynomial} & \textbf{Lower bounds} \\
\hline
$\renprm<1$& $O\left(\frac{\absz^{1/\renprm}}{\esterr^{\max
    (4,2/{\renprm})}}\right)$& &
$O\left(\frac{\absz^{1/\renprm}}{\esterr^{1/\renprm}\log \absz}\right)$ &
for all $\eta>0$, $\Omega(\absz^{1/\renprm-\eta})$\\  
\hline
$\renprm>1$, $\renprm\notin\naturals$&
$O\left(\frac{\absz}{\min(\esterr^{1/{(\renprm-1)}, \esterr^2})}\right)$& & $O\left(\frac{\absz}{\esterr^{1/\renprm}\log
\absz}\right)$ & for all $\eta>0$, $\Omega(\absz^{1-\eta})$\\
\hline
$\renprm>1$, $\renprm\in\naturals$& $O\left(\frac{\absz}{\esterr^2}\right)$& 
$O\left(\frac{\absz^{1-1/\renprm}}{\esterr^2}\right)$&&$\Omega\left(\frac{\absz^{1-1/\renprm}}{\esterr^2}\right)$\\  
\hline
\end{tabular}
\end{center}
\caption{Performance of estimators and lower bounds for estimating \renyi.}
\label{tab:summary}
\end{table}

%%%%%%
\subsection{Organization}
The rest of the paper is organized as follows.
Section~\ref{s:technical} presents basic properties of
power sums of distributions and moments of
Poisson random variables, which may be of independent interest.
The estimation algorithms are analyzed in Section
\ref{s:upper_bounds}, in Section~\ref{sec:empirical} we show results for the empirical or plug-in estimate, in Section~\ref{sec:integral} we
provide optimal results for integral $\renprm$ and finally we provide
an improved estimator for non-integral $\renprm>1$. Examples and simulation of the proposed estimators are given in Section~\ref{s:experiments}.
Section
\ref{s:lower_bounds} contains our lower bounds for
the sample complexity of estimating R\'enyi entropy. Furthermore, in
the Appendix we analyze the performance of the empirical estimator for
power-sum estimation with an additive-accuracy and also derive
lower bounds for its sample complexity.
%Concluding remarks are given in the final section.

\section{Technical preliminaries}\label{s:technical}

\subsection{Bounds on power sums}
Consider a distribution $\dP$ over $[\absz]= \{1,\ldots,\absz\}$. Since R{\'e}nyi 
entropy is a measure of randomness (see \cite{Ren61} for a detailed discussion), it is maximized by the uniform 
distribution and
the following inequalities hold:
\begin{align*}
0 \leq \rental \dP \leq \log \absz, \quad \alpha \neq 1,
\end{align*}
or equivalently
\begin{align}\label{e:renyi_bound} 
1\leq \normP \renprm \leq \absz^{1-\renprm}, \quad \alpha <1 \quad \text{ and }\quad \absz^{1-\renprm} \leq \normP \renprm \leq 1, \quad \alpha >1.
\end{align}
Furthermore, for $\renprm >1$, $\normP{\renprm +\beta}$
and $\normP{\renprm - \beta}$ can be bounded in terms of $\normP \renprm$,
using the monotonicity of norms and of H\"older means (see, for instance, \cite{HarLitPol52}).
\begin{lemma}\label{lem:bnd_moments}
For every $0 \leq  \renprm$, 
\[
\normP{2\renprm} \leq {\normP{\renprm}}^2
\]
Further, for $\renprm > 1$ and $ 0\leq \beta \leq \renprm$,
\[
\normP {\renprm+\beta}\le \absz^{(\renprm-1)(\renprm-\beta)/\renprm}\, \normP {\renprm}^2,
\]
and
\[
\normP {\renprm -\beta}\le \absz^\beta\, \normP {\renprm}.
\]
\end{lemma}
\begin{proof}
By the monotonicity of norms,  
\[
\normP {\renprm+\beta}\le \normP \renprm^{{\frac{\renprm+\beta}{\renprm}}}, 
\]
which gives 
\[
\frac {\normP {\renprm+\beta}}{\normP \renprm^2} \le  {\normP \renprm}^{\frac
  \beta\renprm-1}.
\]
The first inequality follows upon choosing $\beta = \renprm$. For $1 < \renprm$ and $0\leq \beta \leq \renprm$,
we get the second by \eqref{e:renyi_bound}. Note that by the monotonicity of H\"older means, we have
\begin{equation*}
\left(\frac{1}{\absz}\sum_{\smb} \dPx^{\renprm - \beta}\right)^\frac{1}{\renprm -\beta}
\leq \left(\frac{1}{\absz}\sum_{\smb} \dPx^\renprm\right)^\frac{1}{\renprm}.
\end{equation*}
The final inequality follows upon rearranging the terms and using \eqref{e:renyi_bound}.
\end{proof}
%\begin{lemma}
%\label{lem:bnd_moments}
%For every $0 < \renprm \neq 1$, 
%\[
%\normP{2\renprm} \leq {\normP{\renprm}}^2
%\]
%Further, for $\renprm > 1$ and $ 0\leq \beta \leq \renprm$,
%\[
%\normP {\renprm+\beta}\le \absz^{(\renprm-1)(\renprm-\beta)/\renprm}\, \normP {\renprm}^2,
%\]
%and
%\[
%\normP {\renprm -\beta}\le \absz^\beta\, \normP {\renprm}.
%\]
%\end{lemma}
%\noindent A straightforward proof is given in the Appendix.

\subsection{Bounds on moments of a Poisson random variable}
Let $\psnlmb$ be the Poisson distribution with parameter $\lambda$. We consider Poisson sampling where $N \sim \psns n$ samples are drawn from the distribution $\dP$ and the multiplicities used
in the estimation are based on the sequence $\XoN = X_1, ..., X_\Nsmp$ instead of $\Xon$. Under Poisson sampling, the multiplicities $\Mltsmb$ are distributed as $\psns{\nsmp \dPx}$ and are all independent, leading to simpler analysis. To facilitate our analysis under Poisson sampling, we note 
a few properties of the moments of a Poisson random variable.

We start with the expected value and the variance of falling powers of a Poisson random variable. 
\begin{lemma}\label{lem:fall_fac}
Let $X\sim\psns{\lambda}$. Then, for all $r\in\naturals$ 
\[
\expectation{\fpXf}=\lambda^r
\]
and 
\[
\variance{\fpXf}
\le
\lambda^{\flnpwrprm}
\Paren{(\lambda+\flnpwrprm)^{\flnpwrprm}-\lambda^{\flnpwrprm}}.
\]
\end{lemma}
\begin{proof} 
The expectation is
\begin{align*}
\expectation\fpXf
&=
\sum_{i=0}^{\infty}\psnlmbs{i}\cdot\flnpwrss i\flnpwrprm\\
&=
\sum_{i=r}^{\infty}\poiprob{\lambda}{i}\cdot\frac{i!}{(i-r)!}\\
&=
\lambda^{\flnpwrprm}\sum_{i=0}^{\infty}\poiprob{\lambda}{i}\\
%=\lambda^{\flnpwrprm}\sum_{i=0}^{\infty}\poip{\lambda}{i}
&=\lambda^{\flnpwrprm}.
\end{align*}
The variance satisfies
\begin{align*}
\expectation{(\fpXf)^2}
&=
\sum_{i=0}^{\infty}\psnlmbs{i}\cdot(\flnpwrss i\flnpwrprm)^2\\
&=
\sum_{i=r}^{\infty}\poiprob{\lambda}{i}\frac{i!^2}{(i-r)!^2}\\
&=
\lambda^{\flnpwrprm}\sum_{i=0}^{\infty}\poiprob{\lambda}{i}
\cdot\flnpwrss{(i+\flnpwrprm)}\flnpwrprm\\
&=
\lambda^\flnpwrprm\cdot\expectation{\flnpwrss{(X+\flnpwrprm)}\flnpwrprm}
\\
&\le
\lambda^{\flnpwrprm}\cdot
\expectation{\sum_{j=0}^{\flnpwrprm}\binom\flnpwrprm j
\flnpwrss Xj\cdot\flnpwrprm^{\flnpwrprm-j}}\\
&=
\lambda^{\flnpwrprm}\cdot
\sum_{j=0}^{\flnpwrprm}\binom\flnpwrprm j
\cdot\lambda^{j}\cdot\flnpwrprm^{\flnpwrprm-j}\\
&=
\lambda^{\flnpwrprm} (\lambda+\flnpwrprm)^{\flnpwrprm},
\end{align*}
where the inequality follows from
\begin{equation*}
\flnpwrss{(X+\flnpwrprm)}\flnpwrprm
=
\prod_{j=1}^r\Sparen{(X+1-j)+r}
\le
\sum_{j=0}^{\flnpwrprm}\binom\flnpwrprm j
\cdot\flnpwrss Xj\cdot\flnpwrprm^{\flnpwrprm-j}.
\end{equation*}
Therefore,
\[
\variance{\fpXf}
=
\expectation{(\fpXf)^2}-[\,\EE\,\fpXf\,]^2
\le
\lambda^{\flnpwrprm}\cdot
\Paren{(\lambda+\flnpwrprm)^{\flnpwrprm}-\lambda^{\flnpwrprm}}.
\qedhere
\]
\end{proof}

The next result establishes a bound on the moments of a Poisson random variable.
\begin{lemma}\label{l:bound_Poisson_moments} Let $X \sim \psns \lambda$ and 
let $\beta$ be a positive real number. 
Then,
\[
\expectation{X^\beta} \leq 2^{\beta+2} \max\{\lambda, \lambda^\beta\}.
\]
\end{lemma}
\begin{proof}
Let $ Z =  \max\{\lambda^{1/\beta}, \lambda\}$. 
\begin{align*}
 \expectation{\frac{X^\beta}{Z^\beta}}
& \leq \expectation{\left(\frac{X}{Z} \right)^{\lceil \beta \rceil} + \left(\frac{X}{Z} \right)^{\lfloor \beta \rfloor} } \\
& = 
\sum^{\lceil \beta \rceil}_{i=1} 
\left( \frac{\lambda}{Z} \right)^{\lceil \beta \rceil}  
{\lceil \beta \rceil \choose i}
+ 
\sum^{\lfloor \beta \rfloor}_{i=1} 
\left( \frac{\lambda}{Z} \right)^{\lfloor \beta \rfloor}  
{\lfloor \beta \rfloor \choose i} \\
& \leq 
\sum^{\lceil \beta \rceil}_{i=1} 
%\left( \frac{\lambda}{Z} \right)^{\lceil \beta \rceil}  
{\lceil \beta \rceil \choose i}
+ 
\sum^{\lfloor \beta \rfloor}_{i=1} 
{\lfloor \beta \rfloor \choose i} \\
& \leq 2^{\beta + 2}.
\end{align*}
The first inequality follows from the fact that either $X/Z >1 $
or $\leq 1$. The equality follows from the fact that the integer moments of Poisson distribution are Touchard polynomials in $\lambda$.
The second inequality uses the property that $\lambda / Z \leq 1$. Multiplying both sides by $Z^\beta$ results in the lemma.
\end{proof}

We close this section with a bound for  $| \expectation{X^\renprm} - \lambda^\renprm|$, which will be used in the next section and is also of independent interest.

\begin{lemma}\label{l:bound_Poisson_moments2}
For $X\sim\psnlmb$, 
\[
\left|\expectation{X^\renprm} - \lambda^\renprm\right|
\le
\begin{cases}
\renprm \left(2^{\alpha} \lambda + (2^{\alpha}+1) \lambda^{\renprm - 1/2}\right)
& \renprm>1\\
\min\left\{\lambda^\renprm,\lambda^{\renprm-1}\right\}
& \renprm\le1.
\end{cases}
\] 
\end{lemma}
\begin{proof}
For $\renprm \leq 1$, $(1+y)^\renprm \geq 1  +\renprm y - y^2$
for all $y \in [-1,\infty]$, hence,
\begin{align*}
X^\renprm 
&= \lambda^\renprm \left( 1 + \Bigl( \frac{X}{\lambda} - 1 \Bigr)\right)^\alpha \\
& \geq \lambda^\renprm \left( 1 + \alpha \Bigl( \frac{X}{\lambda} - 1 \Bigr) - \Bigl( \frac{X}{\lambda} - 1 \Bigr)^2\right).
\end{align*}
Taking expectations on both sides,
\begin{align*}
\expectation{ X^\renprm }
& \geq  \lambda^\renprm \left( 1 + \alpha \expectation{ \frac{X}{\lambda} - 1 } - \expectation{ \Bigl( \frac{X}{\lambda} - 1 \Bigr)^2}\right)\\
& = \lambda^\renprm \left(2 - \frac{1}{\lambda} \right).
\end{align*}
Since $x^\renprm$ is a concave function and $X$ is nonnegative, the previous
bound yields
\begin{align*}
\left|\expectation{ X^\renprm}  - \lambda^\renprm\right| &= \lambda^\renprm - \expectation{ X^\renprm }
\\
&\leq \min\left\{\lambda^\renprm,\lambda^{\renprm-1}\right\}.
\end{align*}
For $\renprm >1$, 
\[
|x^\renprm - y^\renprm| \leq \renprm |x-y|\left(x^{\renprm-1} + y^{\renprm-1}\right),
\]
hence by the Cauchy-Schwarz Inequality,
\begin{align*}
\expectation{ \left|X^\renprm - \lambda^\renprm\right|} &\leq  \renprm \expectation{ |X -\lambda| \left(X^{\renprm -1} + \lambda^{\renprm -1}\right)}\\
&\leq  \renprm \sqrt{ \expectation{ (X-\lambda)^2}}\sqrt{\expectation{\left(X^{2\renprm -2} + \lambda^{2\renprm -2}\right)}}\\
&\leq \renprm \sqrt{\lambda}\sqrt{\expectation{\left(X^{2\renprm -2} + \lambda^{2\renprm -2}\right)}}\\
&\leq  \renprm \sqrt{2^{2\alpha}\max\{\lambda^2, \lambda^{2\renprm -1} \}  + \lambda^{2\renprm - 1}}\\
&\leq  \renprm \left(2^{\alpha} \max\{\lambda, \lambda^{\renprm -1/2}\} + \lambda^{\renprm - 1/2}\right),
\end{align*}
where the last-but-one inequality is by Lemma \ref{l:bound_Poisson_moments}. 
\end{proof}

\subsection{Polynomial approximation of $x^\renprm$}
\label{sec:poly_approx}
In this section, we review a bound on the error in approximating 
$x^\renprm$ by a $d$-degree polynomial over a bounded interval.
Let $\cP_\degree$ denote the set of all polynomials of degree less than or equal to
$\degree$ over $\mathbb{R}$. For a continuous function $f(x)$ and 
$\lambda>0$,  
let 
\[
\polapxerr(f, [0,\lambda])\ed  \inf_{q\in\cP_\degree} \max_{x\in[0,\lambda]} |q(x)-f(x)|. 
\]
\begin{lemma}[\cite{Timan63}] 
\label{lem:timan}
There is a constant $c'_{\renprm}$ such that for any $\degree>0$, 
\[
\polapxerr(x^{\renprm}, [0,1])\le \frac{c'_\renprm}{\degree^{2\renprm}}.
\]
\end{lemma}
To obtain an estimator which does not require a knowledge of the support size
$\absz$, we seek a polynomial approximation $q_\renprm(x)$ of $x^\renprm$
with $q_\renprm(0) = 0$. Such a polynomial $q_\renprm(x)$ can be obtained by a minor modification
of the polynomial $q'_\renprm(x) = \sum^d_{j=0} q_j x^j$ satisfying the error bound in
Lemma~\ref{lem:timan}.
Specifically, we use the polynomial 
$q_\renprm(x) = q'_\renprm(x)- q_0$ for which the approximation error is bounded as
\begin{align}
\max_{x\in [0,1]} |q_{\renprm}(x) - x^\renprm| 
& \leq |q_0| + \max_{x\in [0,1]} |q'_{\renprm}(x) - x^\renprm|
\nonumber
\\
&  = |q'_\renprm(0) -0^\renprm| + \max_{x\in [0,1]} |q'_{\renprm}(x) - x^\renprm|
\nonumber
\\
& \leq 2\max_{x\in [0,1]} |q'_{\renprm}(x) - x^\renprm| 
\nonumber
\\
& = \frac{2c'_\renprm}{\degree^{2\renprm}} 
\nonumber
\\
& \ed \frac{c_\renprm}{\degree^{2\renprm}}.
\label{e:poly_approx_error}
\end{align}

To bound the variance of the proposed polynomial approximation
estimator, we require  a bound on the absolute values of the
coefficients of $q_\renprm(x)$. The following inequality due to Markov
serves this purpose.
\begin{lemma}[\cite{Markov1892}]
\label{lem:polybound}
Let $p(x) = \sum_{j = 0}^d c_j x^j$ be a degree-$d$ polynomial
so that $| p(x) | \leq 1$ for all $x \in [-1, 1]$. Then 
for all $j = 0, \ldots, m$
\[\max_j |c_j|
\leq (\sqrt{2} + 1)^d.
\] 
\end{lemma}
Since $|x^{\renprm}|\le 1$ for $x\in[0,1]$, the approximation
bound \eqref{e:poly_approx_error} implies 
$|q_\renprm(x)|<1+\frac{c_\renprm}{d^{2\renprm}}$ for all
$x\in[0,1]$. It follows from Lemma~\ref{lem:polybound}  that
\begin{align}
\label{eqn:coeff-bound}
\max_{m}|\coefm|<\left(1+\frac{c_\renprm}{d^{2\renprm}}\right) (\sqrt{2}
+ 1)^d. 
\end{align}

% Here we prove that $O(\absz)$ samples are enough
% to estimate the Renyi Entropy for arbitrary $\renprm$. 

\section{Upper bounds on sample complexity}\label{s:upper_bounds}
In this section, we analyze the performances of the estimators we proposed in 
Section \ref{s:estimators}. 
Our proofs are based on bounding 
the bias and the variance of the estimators under Poisson sampling.
We first describe our general recipe and then analyze the performance of each estimator separately. 

Let $X_1, ..., X_\nsmp$ be $\nsmp$ independent samples drawn from a distribution $\dP$
over $\absz$ symbols. Consider an estimate $\Fnc_\renprm\left(\Xon\right) = \frac1{1-\renprm}
\log \mmntest(\nsmp, \Xon)$ of $\rental \dP$ 
which depends on $\Xon$ only through the multiplicities
and the sample size. 
Here $\mmntest(\nsmp, \Xon)$ is the corresponding estimate
of $\normP \renprm$ -- as discussed in Section~\ref{s:introduction}, small additive error in the estimate $\Fnc_\renprm\left(\Xon\right)$ 
of $\rental \dP$ 
is equivalent to small multiplicative error in the estimate $\mmntest(\nsmp, \Xon)$
of $\normP \renprm$. For simplicity, we analyze a randomized estimator
$\tfnc_\renprm$ described as follows: For $\Nsmp\sim \psns{n/2}$, let
\begin{align*}
\tfnc_\renprm\left(\Xon\right) = \begin{cases}
&\text{constant}, \quad \Nsmp >\nsmp,\\
& \frac{1}{1-\renprm}\log \mmntest(\nsmp/2, \XoN), \quad \Nsmp \leq \nsmp. \end{cases}
\end{align*}
The following reduction to Poisson sampling is well-known.
\begin{lemma}{\bf (Poisson approximation 1)} For $\nsmp \geq 8\log(2/\prerr)$ and $\Nsmp \sim \psns{n/2}$,
\begin{align*}
\bPr{|\rental{\dP} - \tfnc_\renprm\left(\Xon\right)| > \esterr }
\leq \bPr{|\rental{\dP} - \frac{1}{1-\renprm}\log \mmntest(\nsmp/2, \XoN)| > \esterr }
+ \frac{\prerr}{2}.
\end{align*}
\end{lemma}
\noindent It remains to bound the probability on the
right-side above, which can be done provided the bias and
the variance of the estimator are bounded.
\begin{lemma}\label{l:bias-variance-PAC}
For $\Nsmp \sim \psns{n}$, let the power sum estimator $\mmntest =\mmntest(\nsmp, \XoN)$ have
bias and variance satisfying
\begin{align*}
\left| \expectation \mmntest  - \normP \renprm\right| &\leq \frac{\esterr}{2} \normP \renprm,\\
\variance \mmntest &\leq \frac{\esterr^2}{12}  \normP \renprm^2.
\end{align*}
Then, there exists an estimator $\mntestmdn$ that uses $18 \nsmp \log (1/\prerr)$ samples
and ensures
\begin{align*}
\bPr{\left| \mntestmdn - \normP \renprm\right| > \esterr\, \normP \renprm} 
%&\leq \bPr{\left| \mmntest _\renprm - \expectation \mmntest  \right|  + \left|  \normP \renprm  - \expectation \mmntest  \right|> %\esterr\, \normP \renprm } \\
&\leq \ep.
\end{align*}
\end{lemma}
\begin{proof}
By Chebyshev's Inequality
\begin{align*}
\bPr{\left| \mmntest - \normP \renprm\right| > \esterr\, \normP \renprm} 
\leq \bPr{\left| \mmntest  - \expectation \mmntest  \right| > \frac \esterr 2\, \normP \renprm } 
\leq \frac 13.
\end{align*}
To reduce the probability of error to $\prerr$, we use the estimate $\mmntest$
repeatedly for $O(\log(1/\prerr))$ independent samples $\XoN$ and take the estimate 
$\mntestmdn$ to be the \emph{sample median} of the resulting estimates\footnote{This technique is often referred to as the {\it median trick}.}. Specifically, 
let $\mntests 1, ..., \mntests t$ denote $t$-estimates of $\normP \renprm$ obtained
by applying $\mmntest$ to independent sequences $\XoN$, and let
$\indicator_{\cE_i}$ be the indicator 
function of the event $\cE_i = \{| \mntests i - \normP \renprm| > \esterr\, \normP \renprm\}$. By the analysis 
above we have $\expectation{\indicator_{\cE_i}} \leq 1/3$ and hence by Hoeffding's inequality
\[
\bPr{\sum_{i=1}^t \indicator_{\cE_i} > \frac t2} \leq \exp (- t/18).
\]
The claimed bound follows on choosing $t=18\log(1/\prerr)$ and noting that if more than half of 
$\mntests 1, ..., \mntests t$ satisfy $| \mntests i - \normP \renprm| \leq \esterr\, \normP \renprm$, then their median must also satisfy the same condition.
\end{proof}

In the remainder of the
section, we bound the bias and the variance for our estimators when the number of samples $\nsmp$ are of the appropriate order.
Denote by $\estrenemp$, $\estrenunb$, and $\estrenplndt$, respectively, the empirical estimator $\frac 1 {1-\renprm}\log \estmmnemp$, 
the bias-corrected estimator $\frac 1 {1-\renprm}\log \estmmnunb$, and the polynomial approximation estimator
$\frac 1 {1-\renprm}\log \estmmnplndt$. We begin by analyzing the performances of $\estrenemp$ and $\estrenunb$
and build-up on these steps to analyze $\estrenplndt$.

\subsection{Performance of empirical estimator}
\label{sec:empirical}
The empirical estimator was presented in~\eqref{eqn:empirical}. Using 
the Poisson sampling recipe given above, we derive upper bound
for the sample complexity of the empirical estimator by 
bounding its bias and variance.
The resulting bound 
for $\alpha>1$ is given in Theorem~\ref{t:upper_bounds_arbitrary} and
for $\alpha<1$ in Theorem~\ref{t:upper_bounds_alpha_small}. 

\begin{theorem}\label{t:upper_bounds_arbitrary}
For $\renprm>1$, $0<\esterr<1/2$, and $0<\prerr<1$, the estimator $\estrenemp$ satisfies
\[
\sam_\renprm^{\estrenemp}(\absz, \esterr, \prerr) \leq
\Order_\renprm\left(\frac{\absz}{\min(\esterr^{1/(\renprm -1)},\esterr^2
%, \absz^{1/\renprm} \esterr^4
)
}\log
  \frac1\prerr\right),
\]
for all $\absz$ sufficiently large.
\end{theorem}

\begin{proof} 

Denote $\npsmb \ed \nsmp \dPs \smb$. For $\renprm >1$,
we bound the bias of the power sum estimator as follows:
\begin{align}
 \left| \expectation{\frac{\sum_\smb  \Mltsmb^\renprm}{\nsmp^\renprm}} - \normP \renprm\right| &\stackrel {(a)} {\leq}
\frac{1}{\nsmp^\renprm}\sum_\smb \left|\expectation{\Mltsmb^{\renprm}} -\npsmb^\renprm\right| 
\nonumber\\
&\stackrel {(b)} {\leq}  \frac{\renprm}{\nsmp^\renprm} \sum_\smb\left(2^{\renprm}\npsmb+ (2^{\renprm}+1)\npsmb^{\renprm - 1/2}\right)
\nonumber\\
&= \frac{\renprm 2^{\renprm} }{\nsmp^{\renprm-1}} + \frac{\renprm (2^{\renprm}+1)}{\sqrt{\nsmp}} \normP{\renprm-1/2}
\nonumber\\
&\stackrel {(c)} {\leq} \renprm \left( 2^{\renprm}\left(\frac{\absz}{\nsmp}\right)^{\renprm -1} +
  (2^{\renprm}+1)\sqrt{\frac{k}{\nsmp}} \right) \normP \renprm 
\nonumber
\\
& \leq 2 \renprm 2^\renprm \left[
\left(\frac{k}{n}\right)^{\renprm-1}
+\left(\frac{k}{n}\right)^{1/2}\right]\normP \renprm, 
\label{e:central_mean_bound1}
\end{align}
where $(a)$ is from the triangle inequality, $(b)$ from
Lemma~\ref{l:bound_Poisson_moments2}, and $(c)$ follows from 
Lemma \ref{lem:bnd_moments} and \eqref{e:renyi_bound}.
Thus, the bias of the estimator is less than 
$\delta(\renprm -1)\normP \renprm/2$ when 
\[
n \geq k \cdot \left(\frac{8\alpha2^{\alpha}}{\delta (\alpha - 1)}  \right)^{\max(2,1/(\alpha-1))}.
\]
Similarly,  to bound the variance, using independence of
multiplicities:
\begin{align}
  \variance{ \sum_x\frac{\Mltx^\renprm}{\nsmp^\renprm}}
  \nonumber
&=  \frac{1}{\nsmp^{2\renprm}}\sum_x \variance{ \Mltx^\renprm}
  \nonumber\\
&= \frac{1}{\nsmp^{2\renprm}}\sum_x  \expectation{ \Mltx^{2\renprm}}
-\Sparen{\EE \Mltx^\renprm}^2  
  \nonumber\\
&\stackrel {(a)} {\leq} \frac{1}{\nsmp^{2\renprm}} \sum_x  \expectation{ \Mltx^{2\renprm}} - \npsmb^{2\renprm}
  \nonumber\\
%  \displaybreak
& {\leq} \frac{1}{\nsmp^{2\renprm}} \sum_x \left|\expectation{\Mltx^{2\renprm}} - \npsmb^{2\renprm}\right| 
  \nonumber\\
& {\leq}  \frac{2\renprm}{\nsmp^{2\renprm}} \sum_\smb\left(2^{2\renprm}\npsmb+ (2^{2\renprm}+1)\npsmb^{2\renprm - 1/2}\right)
\label{e:bound_var_arbitrary}
\\
&= \frac{2\renprm 2^{2\renprm} }{\nsmp^{2\renprm-1}} + \frac{2\renprm (2^{2\renprm}+1)}{\sqrt{\nsmp}} \normP{2\renprm-1/2}
\nonumber
\\
&\stackrel{(c)}{\leq }2\renprm 2^{2\renprm} \left(\frac{k }{\nsmp}\right)^{2\renprm-1}  \normP{\renprm}^2 
+ 2\renprm( 2^{2\renprm}+1) \left(\frac{k^{\frac{\alpha-1}{\alpha}}}{n} \right)^{1/2}  \normP{\renprm}^2 
\ignore{& \leq 2 \renprm 2^\renprm \left(\frac{k}{n}\right)^{1/2}, 
+ 
2 \renprm 2^\renprm \left(\frac{k}{n}\right)^{\renprm-1}, 
\nonumber
\\
& \stackrel {(c)} {\leq} 2\renprm \left(
  c\left(\frac{\absz}{\nsmp}\right)^{2\renprm -1} +
  (c+1)\sqrt{\frac{k}{\nsmp}} \right) \normP{\renprm}^2 }
\nonumber
\end{align}
$(a)$ is from Jensen's inequality since $z^{\renprm}$ is convex and
$\expectation{\Mltx}=\npsmb$, $(c)$ follows from 
%\footnote{For brevity, the constants in \eqref{e:central_mean_bound1}  and \eqref{e:bound_var_arbitrary}, albeit different, are both  denoted by $c$.} %\eqref{e:central_mean_bound1} and 
Lemma \ref{lem:bnd_moments}. 
Thus, the variance is less than $\delta^2(\alpha-1)^2\normP \renprm^2/12$ when
\[
n \geq k \cdot \max \left\} \left( \frac{48\alpha 2^{2\alpha}}{\delta^2 (\alpha-1)^2}\right)^{1/(2\alpha-1)}, 
 \left( \frac{96\alpha2^{2\alpha}}{k^{1/2\renprm}\delta^2 (\alpha-1)^2}\right)^{2}  \right\}
= k \cdot \left( \frac{48\alpha 2^{2\alpha}}{\delta^2 (\alpha-1)^2}\right)^{1/(2\alpha-1)},
\] 
where the equality holds for $k$ sufficiently large. The theorem follows by using 
Lemma~\ref{l:bias-variance-PAC}. 
\end{proof}

\begin{theorem}\label{t:upper_bounds_alpha_small}
For $\renprm<1$, $\esterr>0$, and $0<\prerr<1$,
the estimator $\estrenemp$ satisfies
\[
\sam_\renprm^{\estrenemp}(\absz, \esterr, \prerr) \leq\Order\left(\frac{\absz^{1/\renprm}}{\esterr^{\, \max\{4,\, 2/\renprm\}}}\log \frac1\prerr\right).
\]
\end{theorem}
\begin{proof}
 For $\renprm <1$, once again we take a recourse to 
Lemma~\ref{l:bound_Poisson_moments2} to bound the bias as follows:
\begin{align*}
 \left|\expectation{ \frac{\sum_\smb  \Mltsmb^\renprm}{\nsmp^\renprm}} - \normP \renprm\right| &\leq 
\frac{1}{\nsmp^\renprm} \sum_\smb \left|\expectation{\Mltsmb^\renprm} - \npsmb^\renprm\right| 
\\
&\leq  \frac{1}{\nsmp^\renprm} \sum_\smb\min \left( \npsmb^{\renprm}, \npsmb^{\renprm-1}\right)
\\
&\leq  \frac{1}{\nsmp^\renprm} \left[\sum_{\smb \notin A} \npsmb^{\renprm} +  \sum_{\smb\in A} \npsmb^{\renprm-1}\right],
\end{align*}
for every subset $A \subset [k]$. Upon choosing $A = \{\smb: \npsmb \geq 1\}$, we get
\begin{align}
 \left|\expectation{ \frac{\sum_\smb  \Mltsmb^\renprm}{\nsmp^\renprm}} - \normP \renprm\right| 
&\leq 2\frac{\absz}{\nsmp^\renprm}
\nonumber
\\
&= 2\left(\frac{k^{1/\renprm}}{\nsmp}\right)^\renprm
\nonumber
\\
&\leq 2\normP \renprm \left(\frac{k^{1/\renprm}}{\nsmp}\right)^\renprm,
\label{e:central_mean_bound2}
\end{align}
where the last inequality uses~\eqref{e:renyi_bound}. For bounding the
variance, note that 
\begin{align}
  \variance{ \sum_\smb\frac{\Mltx^\renprm}{\nsmp^\renprm}}
  \nonumber
&=  \frac{1}{\nsmp^{2\renprm}}\sum_\smb \variance{ \Mltx^\renprm}
  \nonumber\\
&= \frac{1}{\nsmp^{2\renprm}}\sum_\smb  \expectation{ \Mltx^{2\renprm}} -\Sparen{\EE \Mltx^\renprm}^2 
  \nonumber\\
& = \frac{1}{\nsmp^{2\renprm}} \sum_\smb\expectation{\Mltx^{2\renprm} }- \npsmb^{2\renprm} + \frac{1}{\nsmp^{2\renprm}} \sum_\smb \npsmb^{2\renprm} - \Sparen{\EE \Mltx^\renprm}^2.
\label{eqn:var-empirical}
\end{align}
Consider the first term on the right-side. For $\alpha \leq 1/2$, it is bounded above by $0$ since $z^{2\renprm}$ is concave in $z$, and for $\alpha > 1/2$ the bound in \eqref{e:bound_var_arbitrary} and Lemma~\ref{lem:bnd_moments} applies to give
\begin{align}
\frac{1}{\nsmp^{2\renprm}} \sum_x  \expectation{\Mltx^{2\renprm}} -\npsmb^{2\renprm}
 \leq 2\renprm\left( \frac{c }{\nsmp^{2\renprm-1}} +  (c+1)\sqrt{\frac{\absz}{\nsmp}} \right)\normP{\renprm}^2.
\label{eqn:var-empirical1}
\end{align}
For the second term, we have 
\begin{align}
\sum_{\smb}\npsmb^{2\renprm} - \Sparen{\EE \Mltx^\renprm}^2 &= \sum_{\smb}\left(\npsmb^{\renprm} - \expectation{\Mltx^\renprm}\right)\left(\npsmb^{\renprm} + \expectation{ \Mltx^\renprm}\right)
\nonumber
\\
&\stackrel{(a)}{\leq} 2 \nsmp^\renprm\normP \renprm  \left(\frac{k^{1/\renprm}}{\nsmp}\right)^\renprm \sum_{\smb}\left(\npsmb^{\renprm} + \expectation{ \Mltx^\renprm}\right)
\nonumber
\\
&\stackrel{(b)}{\leq} 4\nsmp^{2\renprm}\normP \renprm^2
\left(\frac{k^{1/\renprm}}{\nsmp}\right)^\renprm, 
\nonumber
\end{align}
where $(a)$ is from~\eqref{e:central_mean_bound2}
and $(b)$ from the concavity of $z^\renprm$ in $z$. The
proof is completed by combining the two bounds above and using
Lemma~\ref{l:bias-variance-PAC}. 
\end{proof}

In fact, we show in the appendix that the dependence on $\absz$ implied by the previous two results are optimal.
\begin{theorem}
Given a sufficiently small $\esterr$, the sample complexity 
$S_{\renprm}^{\estrenemp}(\absz,\esterr, \prerr)$
of the empirical estimator $\estrenemp$  is bounded below as
\begin{align*}
S_{\renprm}^{\estrenemp}(\absz,\esterr, 0.9)= \begin{cases}
\Omega\left(\frac{\absz}{\esterr}\right), \quad &\renprm>1,\\
\Omega\left(\frac{\absz^{1/\renprm}}{\esterr^{1/\renprm}}\right), \quad &\renprm<1.
\end{cases}
\end{align*}
\end{theorem}
While the performance of the empirical estimator is limited by these bounds, below we exhibit estimators that beat these bounds and thus outperform the empirical estimator. 

\subsection{Performance of bias-corrected estimator for integral $\renprm$}
\label{sec:integral}
To reduce the sample complexity for integer orders $\renprm>1$ to below $\absz$, 
we follow the development of Shannon entropy estimators.
Shannon entropy was first estimated via an empirical
estimator, analyzed in, for instance,~\cite{AntosK01}.
However, with $o(\absz)$ samples, the bias of the empirical estimator
remains high~\cite{Paninski04b}.
This bias is reduced by the Miller-Madow correction~\cite{Miller95,
Paninski04b}, but even then, $\Order(\absz)$ samples are needed for a
reliable Shannon-entropy estimation~\cite{Paninski04b}.

% (see Lemma~\ref{lem:all_fac}).}
Similarly, we reduce
  the bias for \renyi estimators using \emph{unbiased 
estimators} for $\dPx^\renprm$ for integral $\renprm$.
We first describe our estimator, and in
Theorem~\ref{t:upper_bounds_integer} we show that for $1<\renprm\in\integers$,
$\estmmnunb$ estimates $\normP\renprm$ using
$\Order(\absz^{1-1/\renprm}/\esterr^2)$ samples. 
Theorem~\ref{t:lower_bounds_integer2} in Section~\ref{s:lower_bounds}
shows that this number is optimal up to constant factors.

Consider the unbiased estimator for $\normP\renprm$ given by
\[
\estmmnunb
\ed
\sum_x \frac{\flnpwrss\Mltsmb\renprm}{n^\renprm},
\]
which is unbiased since by Lemma~\ref{lem:fall_fac},
\[
\expectation{
\estmmnunb}
=
\sum_{\smb}
\expectation{\frac{\flnpwrss\Mltsmb\renprm}{\nsmp^\renprm}}
=
\sum_{\smb}
p_\smb^\renprm
=
\normP\renprm.
\]
Our {\em bias-corrected} estimator for $\rental \dP$ is 
\[
{\hat H}_\renprm = \frac 1 {1-\renprm} \log \estmmnunb.
\]

The next result provides a bound for the number of samples 
needed for the bias-corrected
estimator.
\begin{theorem}\label{t:upper_bounds_integer}
For an integer $\renprm>1$,  any $\esterr>0$, and $0<\prerr<1$,
the estimator $\estrenunb$ satisfies
\[
\sam_\renprm^{\estrenunb}(\absz, \esterr, \prerr) \leq O\left(\frac{\absz^{(\renprm-1)/\renprm}}{\esterr^2}\log \frac1\prerr\right).
\]
\end{theorem}
\ignore{
\begin{remark*}
Since we use Poisson sampling to simplify the analysis, we chose to remove the bias of $\Fnc_\alpha$
under Poisson sampling to obtain $\Gnc_\alpha$. Alternatively, we can remove the bias under the usual sampling and define
\[
\Gnc_\renprm(\Xon) = \frac{1}{1-\renprm} \log \sum_{\smb} \frac{\left(\Mltsmb\right)_\renprm}{(\nsmp)_\renprm};
\]
a similar result as Theorem~\ref{t:upper_bounds_integer} can be obtained for this estimator, albeit with a slightly different proof.
\end{remark*}}
\begin{proof}
Since the bias is 0, we only need to bound the variance to use
Lemma~\ref{l:bias-variance-PAC}. 
To that end, we have
\begin{align}
\variance{\frac{{\sum_{\smb}
      \flnpwrrp{\Mltsmb}}}{\nsmp^{\renprm}}}
&= \frac1{n^{2\renprm}}\sum_{\smb}\variance{ \flnpwrrp{\Mltsmb}}
\nonumber
\\
&\le \frac1 {\nsmp^{2\renprm}}\sum_{\smb}\left(\npsmb^{\renprm}(\npsmb+\renprm)^{\renprm}-\npsmb^{2\renprm}\right)
\nonumber
\\
&= \frac1 {\nsmp^{2\renprm}}\sum_{r=0}^{\renprm-1}\sum_{\smb} {\renprm\choose
  r}\renprm^{\renprm-r}{\npsmb}^{\renprm+r}
\nonumber
\\
%&=\frac1 {\nsmp^{2\renprm}}\sum_{r=0}^{\renprm-1}n^{\renprm+r} {\renprm\choose
%  r}\renprm^{\renprm-r}\sum_{\smb}{\dPs \smb}^{\renprm+r}\\
&= \frac1 {\nsmp^{2\renprm}}\sum_{r=0}^{\renprm-1}n^{\renprm+r} {\renprm\choose
  r}\renprm^{\renprm-r}\normP {\renprm+r},
\label{e:variance_integer}
\end{align}
where the inequality uses Lemma \ref{lem:fall_fac}. It follows from Lemma \ref{lem:bnd_moments} that
\begin{align*}
\frac1{n^{2\renprm}}\frac{\variance{ \sum_{\smb}
      \flnpwrrp{\Mltsmb} }}{\normP {\renprm}^2}
&\le \frac1{n^{2\renprm}} \sum_{r=0}^{\renprm-1}n^{\renprm+r} {\renprm\choose
  r}\renprm^{\renprm-r}\frac{\normP {\renprm+r}}{\normP {\renprm}^2}\\
&\le \sum_{r=0}^{\renprm-1}n^{r-\renprm} {\renprm\choose
  r}\renprm^{\renprm-r}\absz^{{(\renprm-1)(\renprm-r)/\renprm}}\\
%& = \sum_{r=0}^{\renprm-1}
%\renprm^{2(\renprm-r)}\left(\frac{\absz^{{(\renprm-1)/\renprm}}}\nsmp\right)^{\renprm-r}\\
& \le \sum_{r=0}^{\renprm-1}\left(\frac{\renprm^2\absz^{{(\renprm-1)/\renprm}}}\nsmp\right)^{\renprm-r},
\end{align*}
which is less than $\esterr^2/12$ if $\renprm^2\absz^{1-1/\renprm}/n$,
for all $\esterr$ sufficiently small.   
%which is further bounded above by $1$ if 
%\[
%\left(\frac{\renprm^2\absz^{{(\renprm-1)/\renprm}}}\nsmp\right) <\frac12.
%\]
Applying Lemma~\ref{l:bias-variance-PAC} completes the proof.
\end{proof}

\subsection{The polynomial approximation estimator} 
\label{sec:nonintegral}
Concurrently with a conference version of this paper \cite{AcharyaOTT15}, a polynomial
approximation based approach was proposed  in~\cite{JiaoVW14} and~\cite{WuY14} 
for estimating \emph{additive functions} of the form $\sum_\smb f(\dPx)$. 
As seen in Theorem~\ref{t:upper_bounds_integer}, 
polynomials of probabilities have succinct unbiased
estimators. Motivated by this observation, instead of estimating $f$,
these papers consider
estimating a polynomial that is a \emph{good approximation} to $f$.  
The underlying heuristic for this approach
is that the difficulty in estimation arises from small probability symbols
since empirical estimation is nearly optimal for symbols with large probabilities.
On the other hand, there is no loss in estimating a polynomial approximation of 
the function of interest for symbols with small probabilities.

In particular, \cite{JiaoVW14} considered the problem of estimating 
power sums $\normP \renprm$ up to additive accuracy and showed that
$\Order\left(\absz^{1/\renprm}/\log \absz\right)$ samples suffice for $\renprm<1$.
Since $\normP \renprm \geq 1$ for $\renprm <1$, this in turn implies
a similar sample complexity for estimating $\rental \dP$ for $\renprm <1$.
On the other hand, $\alpha>1$, the power
sum $\normP\renprm\leq 1$ and can be small ($e.g.$, it is $\absz^{1-\alpha}$ for
the uniform distribution). In fact, we show in the Appendix that 
additive-accuracy estimation of power sum is easy for $\renprm>1$
and has a constant sample complexity.
Therefore, additive guarantees for
estimating the power sums are insufficient to estimate the \renyi. 
Nevertheless, our
analysis of the polynomial estimator below shows that it attains the
$\Order(\log \absz)$ improvement in sample complexity over the empirical estimator
even for the case $\renprm>1$.

We first give a brief description of the polynomial estimator of
\cite{WuY14} 
and then in Theorem~\ref{t:upper_bounds_arbitrary_poly} prove that for
$\renprm>1$ the sample complexity of $\estmmnplndt$ is
$\Order(\absz/\log \absz)$. For completeness, we also include a proof
for the case $\renprm<1$, which is slightly different from the one in \cite{JiaoVW14}.

Let $N_1, N_2$ be independent $\poid{\nsmp}$ random variables. We
consider Poisson sampling with two set of samples drawn from $\dP$,
first of size $N_1$ and the second $N_2$.  Note that the total number
of samples $N=N_1+N_2\sim\poid{2\nsmp}$.  The polynomial approximation
estimator uses different estimators for different estimated values of
symbol probability $\dPx$.  We use the first $N_1$ samples for
comparing the symbol probabilities $\dPx$ with $\thr/\nsmp$ and the
second is used for estimating $\dPx^\renprm$. Specifically, denote by
$\Mltsmb$ and $\Mltsmb^\prime$ the number of appearances of $\smb$ in
the $N_1$ and $N_2$ samples, respectively.  Note that both $\Mltsmb$
and $\Mltsmb^\prime$ have the same distribution
$\poid{\nsmp\dPx}$. Let $\thr$ be a threshold, and $\degree$ be the
degree chosen later. Given a threshold $\thr$, the polynomial
approximation estimator is defined as follows:
\begin{itemize}
\item[]{$\Mltsmb^\prime > \thr$:} For all such symbols, estimate
  $\dPx^\renprm$ using the empirical estimate
  $(\Mltsmb/\nsmp)^{\renprm}$.

\item[]{$\Mltsmb^\prime\le\thr$:} Suppose $\polyx
  =\sum_{m=0}^\degree\coefm x^m$ is the polynomial satisfying
  Lemma~\ref{lem:timan}. Since we expect $\dPx$ to be less than
  $2\thr/\nsmp$ in this case, we estimate $\dPx^\renprm$ using an
  unbiased estimate of\footnote{Note that if $|q(x) - x^\renprm|<\ep$
    for all $x\in [0,1]$, then $|\eta^\renprm q(x/eta) - x^\renprm|<
    \eta^\renprm\ep$ for all $x\in [0,\eta]$.}  $(2\thr/\nsmp)^\renprm
  \polys {\nsmp \dPx/2\thr}$, namely
\[
\left(\sum_{m=0}^\degree \frac{\coefm(2\thr)^{\renprm -m} \flnpwrss
  \Mltsmb m}{\nsmp^\renprm}\right).
\]  
\end{itemize}

Therefore, for a given $\thr$ and $d$ the combined estimator
$\estmmnplndt$ is
\[
\estmmnplndt \ed \sum_{\smb: \Mltsmb^\prime \le
  \thr}\left(\sum_{m=0}^\degree \frac{\coefm(2\thr)^{\renprm -m}
  \flnpwrss \Mltsmb m}{\nsmp^\renprm}\right) + \sum_{\smb
  :\Mltsmb^\prime > \thr} \left(\frac{\Mltsmb}{\nsmp}\right)^\renprm.
\]
Denoting by $\hdPx$ the estimated probability of the symbol $\smb$,
note that the polynomial approximation estimator relies on the
empirical estimator when $\hdPx>\thr/\nsmp$ and uses the the
bias-corrected estimator for estimating each term in the polynomial
approximation of $\dPx^\renprm$ when $\hdPx\le\thr/\nsmp$.

We derive upper bounds for the sample complexity of the polynomial
approximation estimator.
\begin{theorem}\label{t:upper_bounds_arbitrary_poly}
For $\renprm>1$, $\esterr>0$, $0<\prerr<1$, there exist constants
$c_1$ and $c_2$ such that the estimator $\estmmnplndt$ with $\thr =
c_1\log \nsmp$ and $\degree=c_2\log\nsmp$ satisfies
\[
\sam_\renprm^{\estmmnplndt}(\absz, \esterr, \prerr) \leq
\Order\left(\frac {\absz} { \log\absz}\frac{\log
  (1/\prerr)}{\esterr^{1/\renprm}}\right).
\]
\end{theorem}
\begin{proof} 
We follow the approach in \cite{WuY14} closely.  Choose
$\tau={c^*}{\log \nsmp}$ such that with probability at least
$1-\prerr$ the events $\Mltsmb^\prime>\tau$ and
$\Mltsmb^\prime\le\tau$ do not occur for all symbols $\smb$ satisfying
$\dPx\leq \tau/(2n)$ and $\dPx>2\tau/n$, respectively.  Or
equivalently, with probability at least $1-\prerr$ all symbols $\smb$
such that $\Mltsmb^\prime>\tau$ satisfy $\dPx>\tau/(2n)$ and all
symbols such that $\Mltsmb^\prime\le\tau$ satisfy $\dPx\le2\tau/n$.
We condition on this event throughout the proof.  For concreteness, we
choose $c^* = 4$, which is a valid choice for $n>20\log(1/\prerr)$ by
the Poisson tail bound and the union bound.

Let $\polyx = \sum_{m=0}^\degree \coefm x^m$ satisfy the polynomial
approximation error bound guaranteed by Lemma~\ref{lem:timan}, $i.e.$,
\begin{align}
\label{e:q_approx_bound}
\max_{x\in(0,1)} |\polyx -
x^{\renprm}|<{c_\renprm}/{\degree^{2\renprm}}
\end{align}
To bound the bias of $\estmmnplndt$, note first that for
$\Mltsmb^\prime<\thr$ (assuming $\dPx\le 2\thr/nsmp$ and estimating
$(2\thr/\nsmp)^\renprm q(n\dPx/2\thr)$)
\begin{align}
\left|\expectation{\sum_{m=0}^\degree \frac{\coefm(2\thr)^{\renprm -m}
    \flnpwrss \Mltsmb m}{\nsmp^\renprm}} - \dPx^\renprm\right| &=
\left|\sum_{m=0}^\degree \coefm
\left(\frac{2\thr}{\nsmp}\right)^{\renprm -m} \dPx^m -
\dPx^\renprm\right| \nonumber\\ &=
\frac{(2\thr)^{\renprm}}{\nsmp^{\renprm}}\left|\sum_{m=0}^\degree
\coefm \left(\frac{\nsmp\dPx}{2\thr}\right)^m -
\left(\frac{\nsmp\dPx}{2\thr}\right)^\renprm\right|\nonumber\\ & =
\frac{(2\thr)^{\renprm}}{\nsmp^{\renprm}}\left|q\left(\frac{\nsmp\dPx}{2\thr}\right)-
\left(\frac{\nsmp\dPx}{2\thr}\right)^\renprm\right|\nonumber\\ &<
\frac{(2\thr)^\renprm
  c_\renprm}{(n\degree^2)^\renprm}, \label{eqn:bias-poly-small}
\end{align}
where the last inequality uses~\eqref{e:q_approx_bound} and
$\nsmp\dPx/(2\thr)\le 1$.

For $\Mltsmb^\prime>\thr$, the bias of empirical part of the power sum
is bounded as
\begin{align*}
\left|\expectation{\left(\frac{\Mltsmb}{\nsmp}\right)^{\renprm}} -
\dPx^\renprm\right| & \stackrel{(a)}{\le} \renprm c\frac{
  \dPx}{\nsmp^{\renprm-1}} + \renprm(c+1)\frac{\dPx^{\renprm -\frac
    12}}{\sqrt \nsmp}\\ &\stackrel{(b)}{\le} \renprm c\frac
    {\dPx^\renprm} {(\thr/2)^{\renprm-1}} +
    \renprm(c+1)\frac{\dPx^{\renprm}}{\sqrt {\thr/2}},
\end{align*}
and $(a)$ is from Lemma~\ref{l:bound_Poisson_moments2} and $(b)$ from
$\dPx>\thr/(2\nsmp)$, which holds when $\Mltsmb^\prime>\thr$.  Thus,
by using the triangle inequality and applying the bounds above to each
term, we obtain the following bound on the bias of $\estmmnplndt$:
\begin{align}
\left|\expectation{\mmntest} - \normP\renprm\right| & \le
\frac{\absz(2\thr)^\renprm c_\renprm}{(n\degree^2)^\renprm} +
\renprm\normP \renprm\left[ \frac c{{(\thr/2)}^{\renprm-1}} + \frac
  {c+1}{\sqrt{\thr/2}}\right]\nonumber \\ &\le \normP
\renprm\left[c_\renprm\left(\frac{k\cdot2\thr}{\nsmp\degree^2}\right)^\renprm
  + \frac {\renprm c}{{(\thr/2)}^{\renprm-1}} + \frac
  {\renprm(c+1)}{\sqrt{\thr/2}}\right]
\label{eqn:bias},
\end{align}
where the last inequality uses~\eqref{e:renyi_bound}.

For variance, independence of multiplicities under Poisson sampling
gives
\begin{align}
\label{eqn:var-poly}
\variance{\mmntest} = \sum_{\smb: \Mltsmb^\prime \le
  \thr}\Var\left(\sum_{m=0}^\degree \frac{\coefm(2\thr)^{\renprm -m}
  \flnpwrss {\Mltsmb} m}{\nsmp^\renprm}\right)\ + \sum_{\smb :\Mltsmb
  > \thr} \Var\left(\frac{\Mltsmb}{\nsmp}\right)^\renprm.
\end{align}
Let $a=\max_m |a_m|$.  By Lemma~\ref{lem:fall_fac}, for any $\smb$
with $\dPx\le 2\thr/n$,
\begin{align}
\Var\left(\sum_{m=0}^\degree \frac{\coefm(2\thr)^{\renprm -m}
  \flnpwrss \Mltsmb m}{\nsmp^\renprm}\right) &\le
a^2\degree^2\max_{1\le m\le d} \left\{
\frac{(2\thr)^{2\renprm-2m}}{\nsmp^{2\renprm}}{\Var \flnpwrss \Mltsmb
  m}\right\}\nonumber\\ &\stackrel{(a)}{\le}a^2\degree^2\max_{1\le
  m\le d} \left\{
\frac{(2\thr)^{2\renprm-2m}}{\nsmp^{2\renprm}}(\nsmp\dPx)^m
((\nsmp\dPx+m)^m -\nsmp\dPx^m) \right\}\nonumber
\\ &\stackrel{(b)}{\le} \frac{a^2\degree^2
  (2\thr+\degree)^{2\renprm}}{\nsmp^{2\renprm}},\label{eqn:variance-poly-small}
\end{align}
where $(a)$ is from Lemma~\ref{lem:fall_fac}, and $(b)$ from plugging
$\nsmp\dPx\le 2\thr$.  Furthermore, using similar steps
as~\eqref{e:bound_var_arbitrary} together with
Lemma~\ref{l:bound_Poisson_moments2}, for $\smb$ with $\dPx >
\thr/(2\nsmp)$ we get
\[
\variance{ \left(\frac{\Mltsmb}{\nsmp}\right)^\renprm} \le 2\renprm
c\frac {\dPx^{2\renprm}} {(\thr/2)^{2\renprm-1}} +
2\renprm(c+1)\frac{\dPx^{2\renprm}}{\sqrt {\thr/2}}.
\] 
The two bounds above along with Lemma~\ref{lem:bnd_moments}
and~\eqref{e:renyi_bound} yield
\begin{align}
\label{eqn:var}
\variance{\mmntest} &\le \normP\renprm^2\left[\frac{a^2 \degree^2
    (2\thr+\degree)^{2\renprm}}\nsmp\left(\frac{\absz}{\nsmp}\right)^{2\renprm
    -1} + \frac{2\renprm c} {(\thr/2)^{2\renprm-1}} +
  \frac{2\renprm(c+1)}{\sqrt {\thr/2}} \right].
\end{align}

For $d=\thr/8=\frac12\log\nsmp$, the last terms in~\eqref{eqn:bias}
are $o(1)$ which gives
\[
\left|\expectation{\mmntest} - \normP\renprm\right|
=\normP\renprm\left(c_{\alpha}\left(\frac{32\absz}{(\nsmp \log
  \nsmp)}\right)^\renprm+o(1)\right) .
\]

Recall from~\eqref{eqn:coeff-bound} that
$a<(1+c_{\renprm}/\degree^{2\renprm})(\sqrt 2+1)^\degree$, and
therefore, $a^2=O((\sqrt{2}+1)^{\log n})=n^{c_0}$ for some
$c_0<1$. Using \eqref{eqn:var} we get
\[
\variance{\mmntest} = \Order\left(\normP
\renprm^2\frac{n^{c_0}\log^{2\renprm+2}\nsmp}{\nsmp}\left(\frac{\absz}{\nsmp}\right)^{2\renprm
  -1}\right).
\]
Therefore, the result follows from Lemma~\ref{l:bias-variance-PAC} for
$k$ sufficiently large.
\end{proof}

We now prove an analogous result for $\renprm<1$. 
\begin{theorem}\label{t:upper_bounds_arbitrary_poly_small}
For $\renprm<1$, $\esterr>0$, $0<\prerr<1$, there exist constants
$c_1$ and  $c_2$ such that the estimator $\estmmnplndt$  
with $\thr = c_1\log \nsmp$ and $\degree=c_2\log\nsmp$ satisfies
\[
\sam_\renprm^{\estmmnplndt}(\absz, \esterr, \prerr) \leq
\Order\left(\frac {\absz^{1/\renprm}} { \log\absz}\frac{\log
    (1/\prerr)}{\renprm^2\esterr^{1/\renprm}}\right). 
\]
\end{theorem}

\begin{proof}
We proceed as in the previous proof and set $\thr$ to be $4\log n$. 
The contribution to the bias of the estimator 
for a symbol $\smb$
with $\Mltsmb^\prime<\thr$ remains bounded as in~\eqref{eqn:bias-poly-small}.
For a symbol $\smb$ with $\Mltsmb^\prime>\thr$, 
the bias contribution of the empirical estimator
is bounded as 
\begin{align}
\left|\expectation{\left(\frac{\Mltsmb}{\nsmp}\right)^{\renprm}} -
  \dPx^\renprm\right| &
\stackrel{(a)}{\le} \frac{\dPx^{\renprm-1}}{\nsmp}\nonumber\\
&\stackrel{(b)}{\le} \frac {2\dPx^\renprm} {\thr},
\nonumber
\end{align}
where $(a)$ is by Lemma~\ref{l:bound_Poisson_moments2} and $(b)$ 
uses $\dPx>\thr/(2\nsmp)$, which holds if $\Mltsmb^\prime>\thr$.
Thus, we obtain the following bound on the
bias of $\estmmnplndt$: 
\begin{align}
\left|\expectation{\mmntest} - \normP\renprm\right| &
{\le}
\frac{\absz(2\thr)^\renprm c_\renprm}{(n\degree^2)^\renprm}  
+ \frac 2{\thr}\normP \renprm
\nonumber 
\\
&{\le} \normP
\renprm\left[c_\renprm\left(\frac{k^{1/\renprm}\cdot2\thr}{\nsmp\degree^2}\right)^\renprm
+ \frac 2 \thr
\right],
\nonumber
\end{align}
where the last inequaliy is by~\eqref{e:renyi_bound}.

To bound the variance, first note that bound 
\eqref{eqn:variance-poly-small} still holds for $\dPx \leq 2\thr/\nsmp$.
To bound the contribution to the variance from the terms with $\nsmp\dPx>\thr/2$,
we borrow steps from the
proof of Theorem~\ref{t:upper_bounds_alpha_small}. In
particular,~\eqref{eqn:var-empirical} gives
\begin{align}
\variance{ \sum_{\smb:\Mltsmb^\prime>\thr}\frac{\Mltx^\renprm}{\nsmp^\renprm}}
\leq \frac{1}{\nsmp^{2\renprm}}
\sum_{\smb:\Mltsmb^\prime>\thr}\expectation{\Mltx^{2\renprm} 
}- \npsmb^{2\renprm} + \frac{1}{\nsmp^{2\renprm}} \sum_{\smb:\Mltsmb^\prime>\thr}
\npsmb^{2\renprm} - \Sparen{\EE \Mltx^\renprm}^2. 
\end{align}
The first term can be bounded in the manner of \eqref{eqn:var-empirical1}
as
\begin{align}
\frac{1}{\nsmp^{2\renprm}} \sum_{{\smb:\Mltsmb^\prime>\thr}}
 \expectation{\Mltx^{2\renprm}} -\npsmb^{2\renprm} 
{\leq}2\renprm\left( \frac{c }{\nsmp^{2\renprm-1}} +
 (c+1)\frac{1}{\sqrt{\thr/2}} \right)\normP\renprm^2,
\nonumber
\end{align}
For the second term, we have 
\begin{align}
\frac{1}{\nsmp^{2\renprm}}\sum_{{\smb:\Mltsmb^\prime>\thr}}\npsmb^{2\renprm}
- \Sparen{\EE \Mltx^\renprm}^2 
&=
\frac{1}{\nsmp^{2\renprm}}\sum_{\smb:\Mltsmb^\prime>\thr}
\left(\npsmb^{\renprm} -
\expectation{\Mltx^\renprm}\right)\left(\npsmb^{\renprm} +
\expectation{ \Mltx^\renprm}\right)\nonumber\\
&\stackrel{(a)}{\leq} 
\frac{1}{\nsmp^{2\renprm}}\sum_{\smb:\Mltsmb^\prime>\thr}\left(\npsmb^{\renprm-1}\right)\left(2\npsmb^{\renprm}\right) 
\nonumber
\\
&{=} 2 \sum_{\smb:\Mltsmb^\prime>\thr} \frac{\dPx^{2\renprm}}{\nsmp\dPx}
\nonumber
\\
&\stackrel{(b)}{\leq} \frac 4{\thr} \normP\renprm^2, 
\nonumber
\end{align}
where $(a)$ follows from Lemma~\ref{l:bound_Poisson_moments2}
and concavity of $z^\renprm$ in $z$ and $(b)$ from
$\nsmp\dPx>\thr/2$ and Lemma~\ref{lem:bnd_moments}. 

Thus, the contribution of the terms corresponding to 
$\Mltsmb^\prime>\thr$ in the bias and the variance are  
 $\normP\renprm\cdot o(1)$ and 
$ \normP\renprm^2\cdot o(1)$, respectively, and can be ignored. 
%Considering the
%contribution of the first term, we again choose $d=\frac\renprm2\log
%n$ to get,  
%\[
%\left|\expectation{\mmntest} - \normP\renprm\right|
%=\normP\renprm\left(c_{\alpha}\left(\frac{32\absz^{1/\renprm}}{\renprm^2\nsmp \log
%      \nsmp}\right)^\renprm+o(1)\right) .  
%\]
%For this choice, we again invoke the bound on the coefficient $a$
%to obtain that $a^2 = O(n^{\renprm c_0})$ for some $c_0<1$. 
Choosing $d = \frac \renprm 2 \log n$ and combining the observations above,
we get the following bound for the bias:
\[
\left|\expectation{\mmntest} - \normP\renprm\right|
=\normP\renprm\left(c_{\alpha}\left(\frac{32\absz^{1/\renprm}}{\nsmp\log
      \nsmp \renprm^2}\right)^\renprm+o(1)\right),
\]
and, using~\eqref{eqn:variance-poly-small}, the following bound for the variance:
\begin{align}
\variance{\mmntest} 
&\le \absz \frac{a^2\degree^2
  (2\thr+\degree)^{2\renprm}}{\nsmp^{2\renprm}}\nonumber + \normP \renprm^2\cdot \order(1)\\
&
\le \normP \renprm^2\left[ \left(\frac {a^2}{\nsmp^{\renprm}}\right)(9\log
  \nsmp)^{2\renprm+2}\left(\frac{\absz^{1/\renprm}}{n}\right)^{\renprm}
+ \order(1)\right]
\nonumber 
\end{align}
Here $a^2$ is the largest squared coefficient of the approximating
polynomial and, by~\eqref{eqn:coeff-bound}, is $\Order(2^{2c_0 d}) = \Order(n^{c_0\renprm})$ for some $c_0<1$. 
Thus,
$a^2 = \order(\nsmp^\renprm)$ and the proof follows by Lemma~\ref{l:bias-variance-PAC}.
\end{proof}

\section{Examples and experiments}\label{s:experiments}
We begin by computing  \renyi for uniform and Zipf distributions; the latter example illustrates the lack of uniform continuity of $\rental \dP$ in $\renprm$.
\begin{example}
The \emph{uniform distribution} $U_\absz$ over $[\absz]=\Sets{1\upto\absz}$
is given by
\[
p_i=\frac1\absz\text{\quad for }i\in[\absz].
\]
Its \renyi for every order $1\ne\renprm\ge0$, and hence for all $\renprm\ge0$, is
\[
H_\renprm(U_\absz)
=
\frac1{1-\renprm}\log\sum_{i=1}^\absz\frac 1{\absz^\renprm}
=
\frac1{1-\renprm}\log \absz^{1-\renprm}
=
\log \absz.\hfill
\]
%Figure~\ref{f:estimation_Halpha_uniform} shows the performance of the bias-corrected and the empirical estimators for samples drawn from a uniform distribution.\qed
\end{example}
\begin{example}
\label{exm:zpf_rny}
The \emph{Zipf distribution} $\zipfbk$ 
for $\beta > 0$ and $\absz\in[\absz]$ is given by
\[
\dP_i
=
\frac{i^{-\beta}}{\sum_{j=1}^\absz j^{-\beta}}
\quad\text{for } i\in[\absz].
\]
Its R\'enyi entropy of order $\renprm\ne1$ is
\[
\rental \zipfbk = \frac{1}{1-\renprm}\log \sum_{i=1}^\absz i^{-\renprm\beta} - \frac{\renprm}{1-\renprm}\log \sum_{i=1}^\absz i^{-\beta}.
\]
Table~\ref{tab:Zipf_entropy} summarizes the leading term $g(k)$ in the approximation\footnote{We say $f(n) \sim g(n)$ to denote $\lim_{n\rightarrow \infty} f(n)/g(n) =1$.}  
$\rental \zipfbk \sim g(k)$.

\begin{table}[H]
\begin{center}
\begin{tabular}{|l|l|l|l|}
\hline
 &$\beta<1$ & $\beta=1$   & $\beta>1$\\
\hline
	 $\renprm \beta<1$  &$ \log \absz$    &$ \frac{1-\renprm \beta}{1-\renprm} \log \absz$ &$ \frac{1-\renprm \beta}{1-\renprm} \log \absz$ \\
 \hline
 	 $\renprm \beta=1$  &$ \frac{\renprm- \renprm \beta}{\renprm - 1}\log \absz $  &  $\frac{1}{2}\log \absz$ &$\frac{1}{1 -\renprm} \log \log \absz$\\
\hline
	 $\renprm \beta>1$  &$\frac{\renprm- \renprm \beta}{\renprm - 1}\log \absz$   &$\frac{\renprm}{\renprm -1} \log \log \absz$&$ \text{constant}$\\
 \hline
\end{tabular}
\\
\caption{The leading terms $g(k)$ in the approximations $\rental \zipfbk \sim g(k)$ for different values of $\renprm\beta$ and $\beta$. The case $\renprm\beta =1$ and $\beta =1$ corresponds to the Shannon entropy of $\zipfsk 1$.}
\label{tab:Zipf_entropy}
\end{center}
\end{table}
In particular,  for $\renprm>1$ 
\[
H_\renprm(\zipfsk 1) = \frac{\renprm}{1-\renprm}\log\log \absz + \Theta\left(\frac1{\absz^{ \renprm -1}}\right) + c(\renprm),
\] 
and the difference 
$|H_2(\dP) - H_{2+\ep}(\dP)|$ is  
$\Order\Paren{\ep\log\log \absz}$. Therefore, even for very small $\ep$ this difference is unbounded and approaches infinity in the limit as $\absz$ goes to infinity. \end{example}

We now illustrate the performance of the proposed estimators for various
distributions for $\renprm =2$ in Figures~\ref{fig:estimation_Halpha_2}
and $\renprm =1.5$ in Figures~\ref{fig:estimation_Halpha_15}.  For $\renprm = 2$,
we compare the performance of bias-corrected and empirical estimators.
For $\renprm = 1.5$, we compare the performance of the 
polynomial-approximation  and the empirical estimator. 
For the polynomial-approximation estimator, the threshold 
$\tau$ is chosen  as $\tau= \ln (n)$ and the approximating polynomial degree is chosen as $d = \lceil 1.5 \tau \rceil$.

We test the performance of these estimators over six different distributions: 
the uniform distribution, a step distribution with half of the symbols having
probability $1/(2k)$ and the other half have probability $3/(2k)$, Zipf
distribution with parameter $3/4$ ($p_i \propto i^{-3/4}$), Zipf
distribution with parameter $1/2$ ($p_i \propto i^{-1/2}$), a
randomly generated distribution using the uniform prior on the probability simplex, and another one generated using the Dirichlet-$1/2$ prior.

In both the figures the true value is shown in black and the
estimated values are color-coded, with the solid line representing their
mean estimate and the shaded area corresponding to one standard
deviation.  As expected, bias-corrected estimators outperform
empirical estimators for $\renprm = 2$ and polynomial-approximation
estimators perform better than empirical estimators for $\renprm =
1.5$.
%For $\renprm = 2$ the bias-corrected estimators works quite well for
%$\nsmp=\sqrt \absz$
%and requires roughly $\absz$ samples to work well
%for $\renprm =1.5$.  Note that the empirical estimator is negatively
%biased for $\renprm>1$ and the figures above confirm this.

\begin{figure}[H]
\centering     %%% not \center
\subfloat[Uniform]{\label{fig:uniform}\includegraphics[width=45mm]{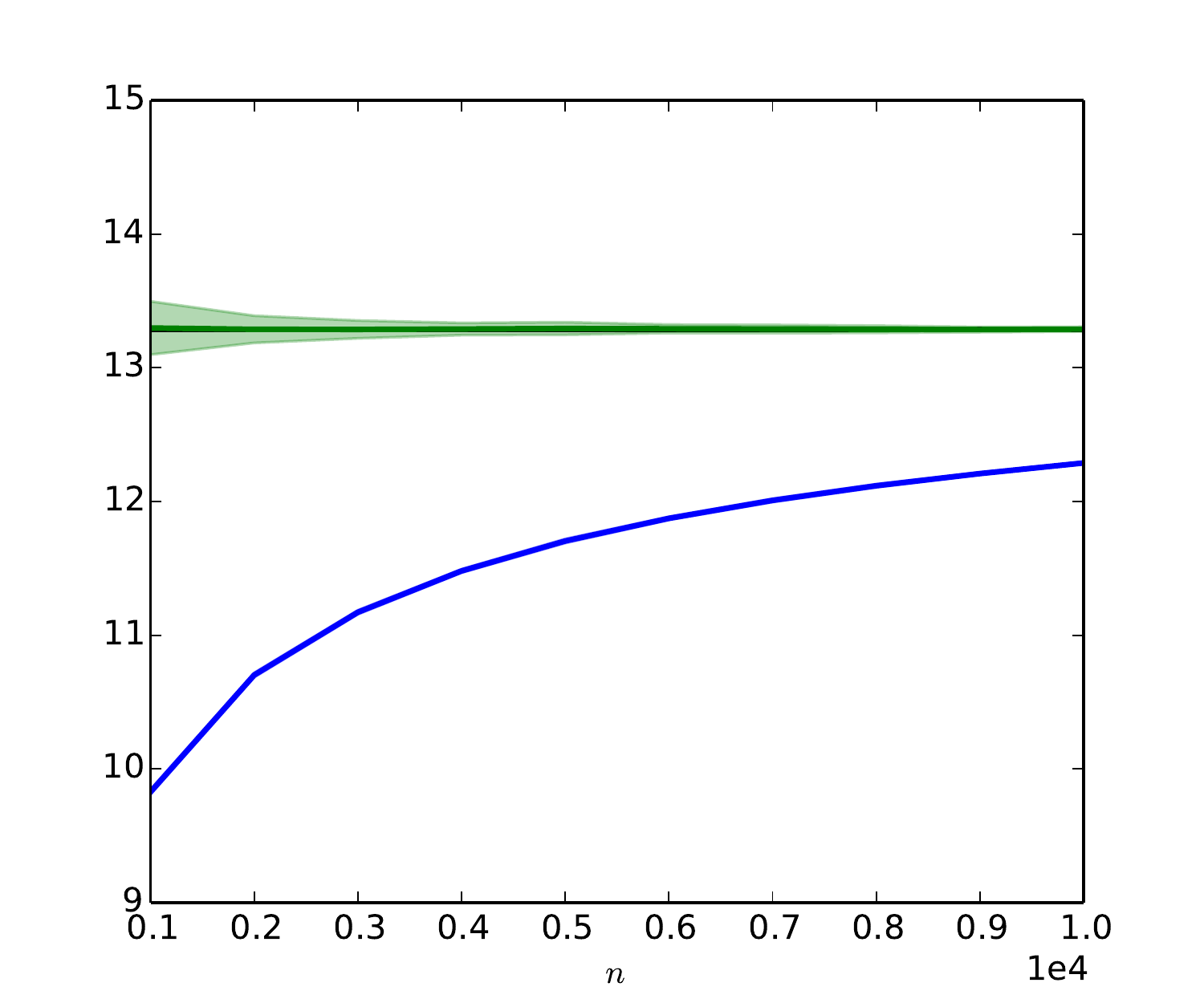}}
\subfloat[Step]{\label{fig:step}\includegraphics[width=45mm]{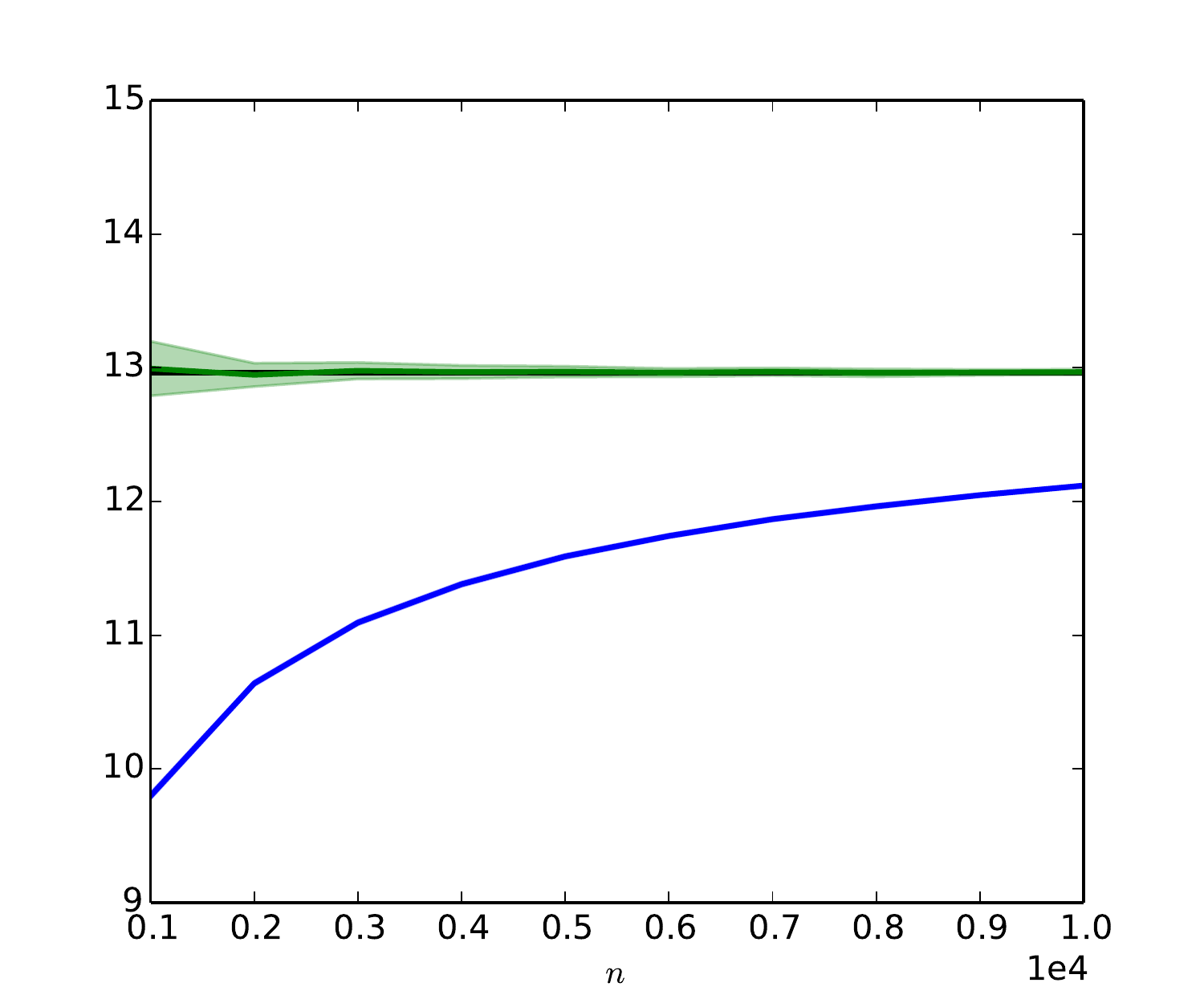}}
\subfloat[Zipf with parameter $3/4$]{\label{fig:zipf1}\includegraphics[width=45mm]{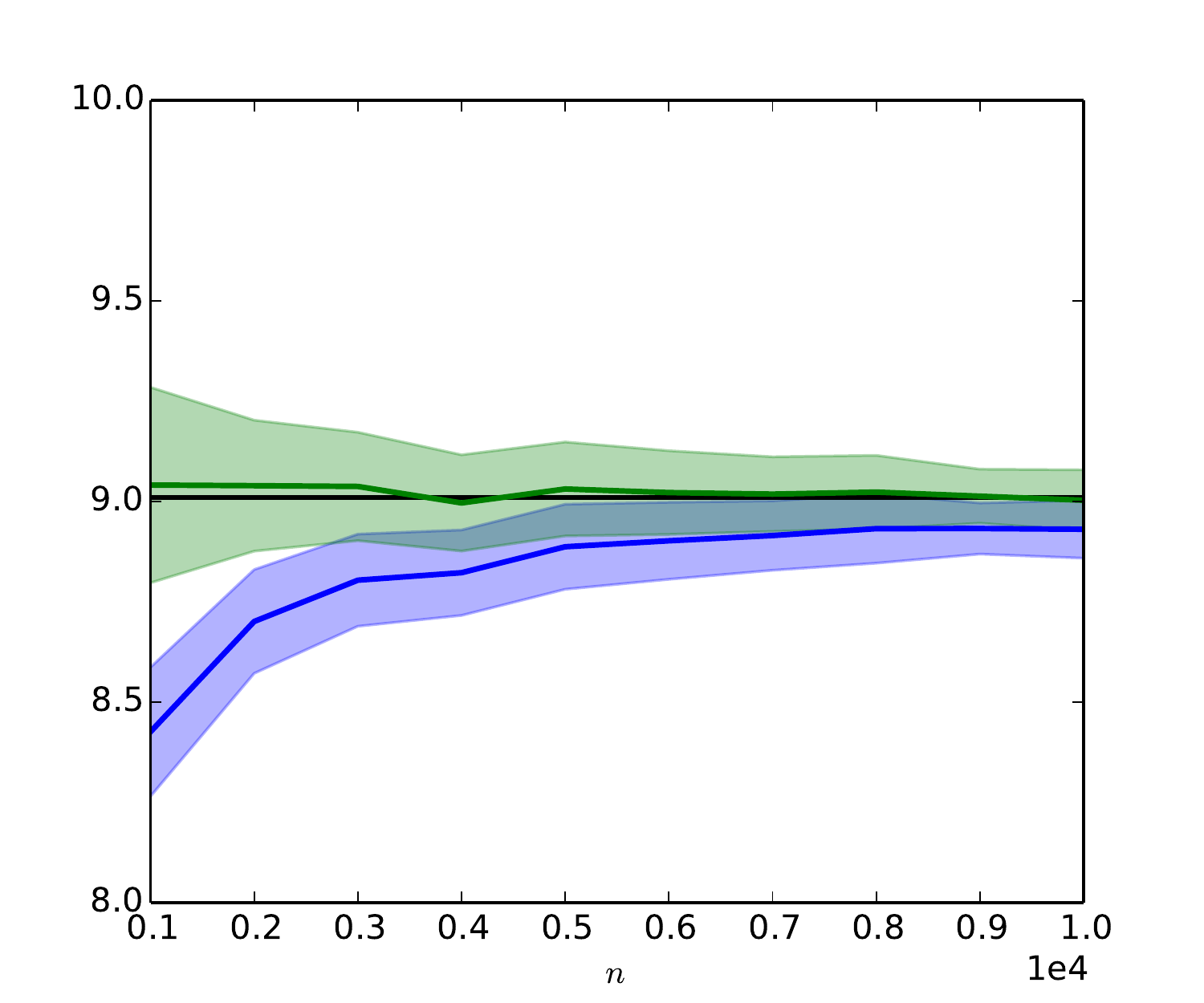}}
\end{figure}
\begin{figure}[H]
%\ContinuedFloat
\centering
\subfloat[Zipf with parameter $1/2$]{\label{fig:zipf15}\includegraphics[width=45mm]{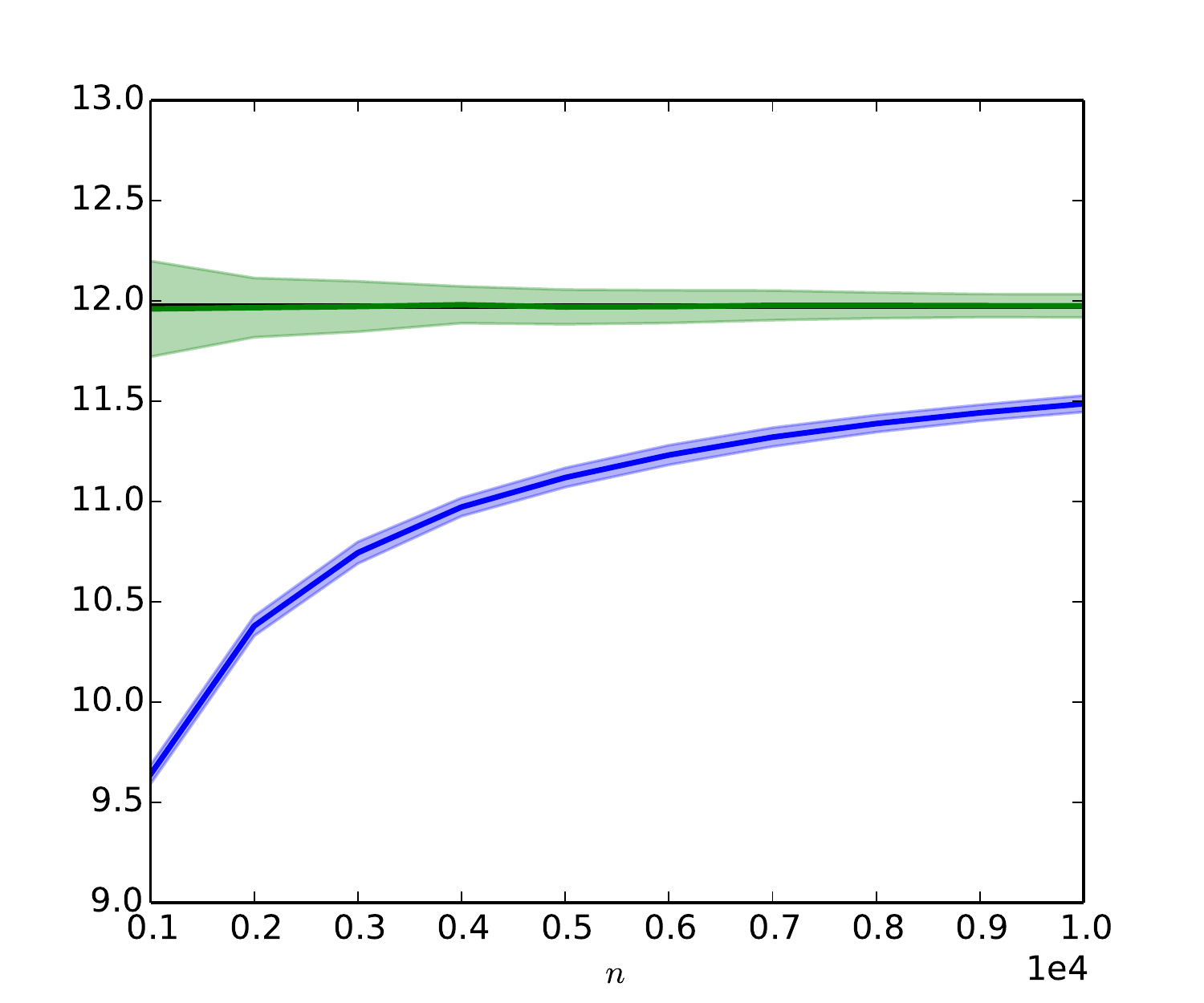}}
\subfloat[Uniform prior (Dirichlet $1$)]{\label{fig:random}\includegraphics[width=45mm]{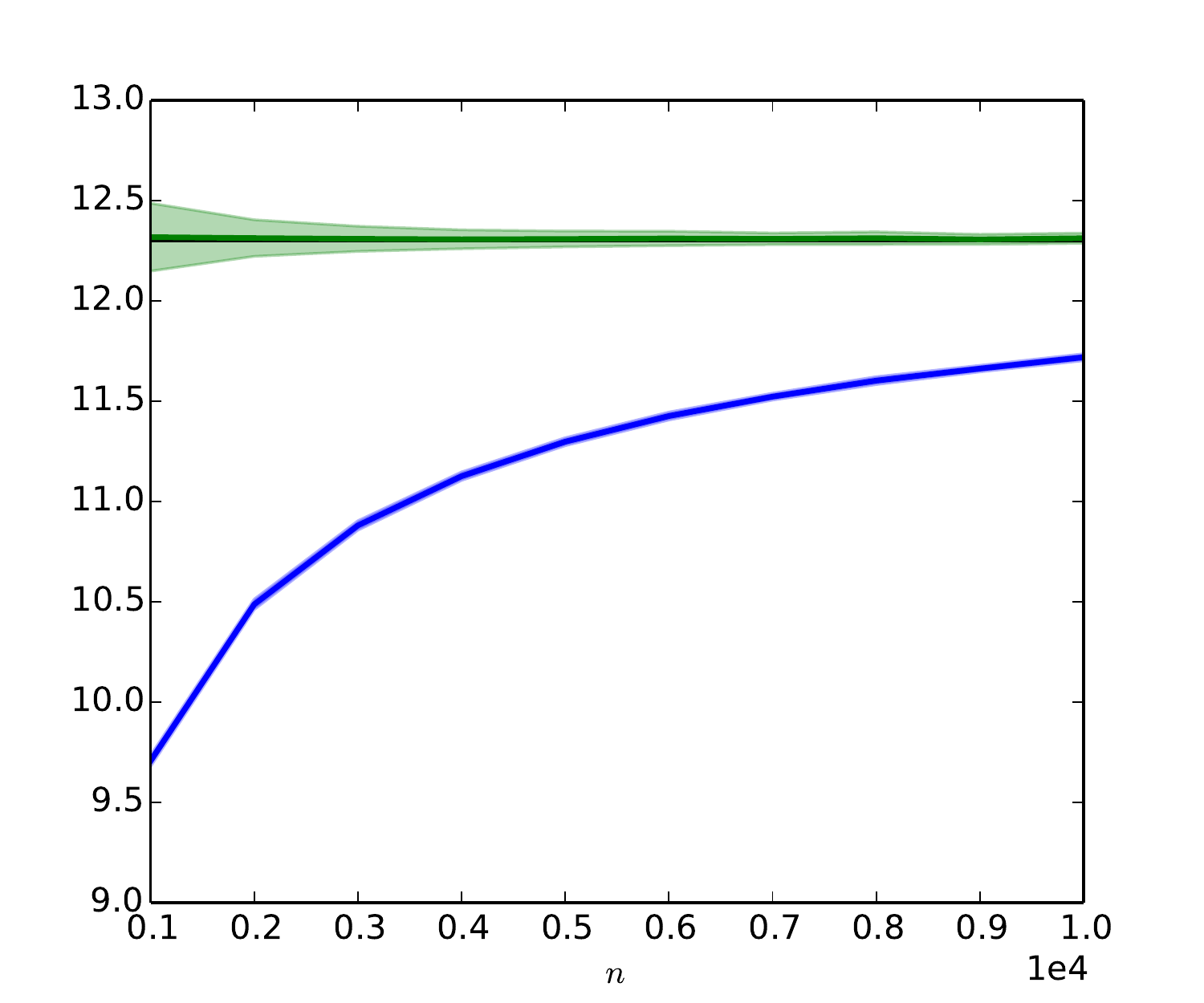}}
\subfloat[Dirichlet $1/2$ prior]{\label{fig:gammahalf}\includegraphics[width=45mm]{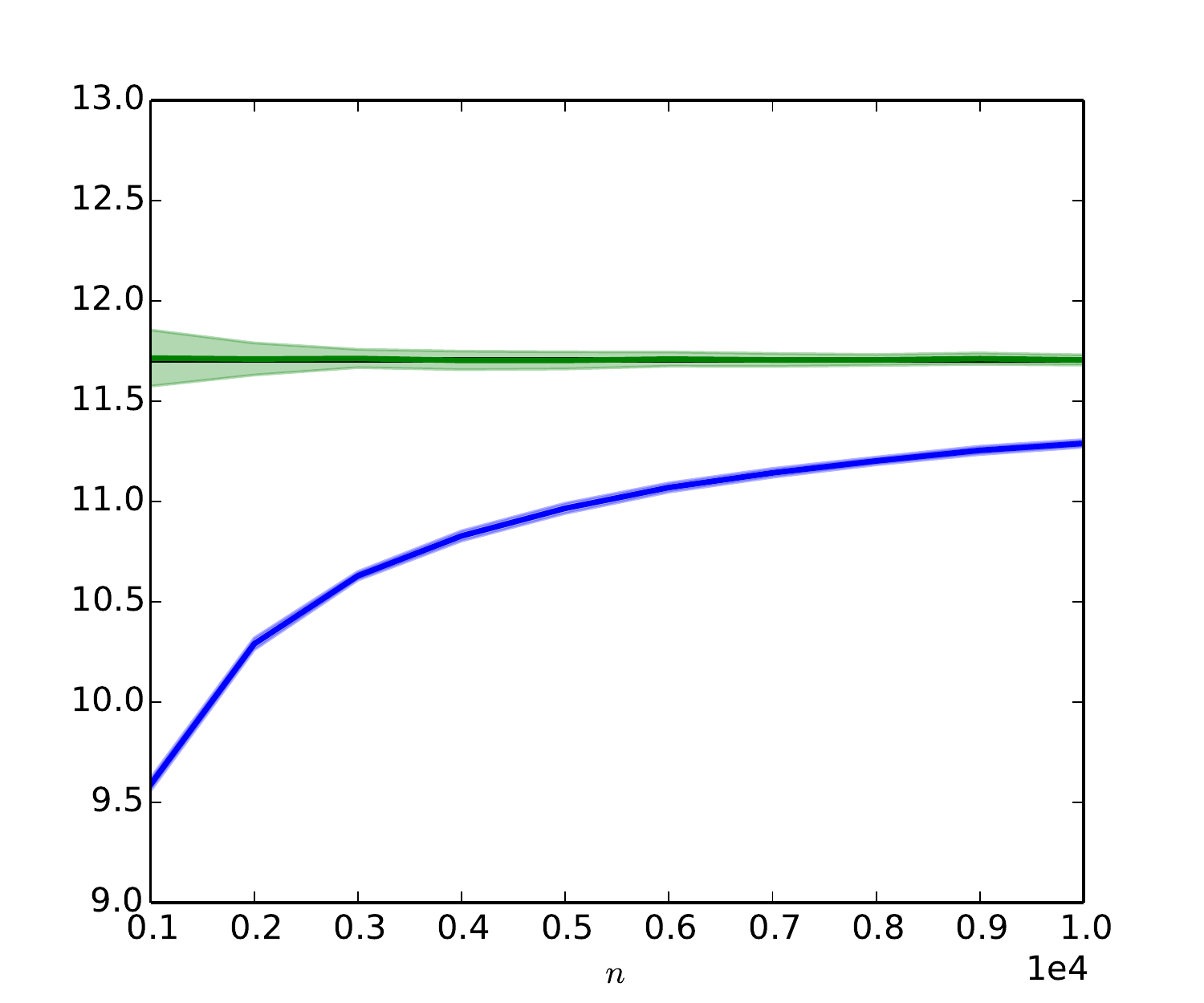}}
\end{figure}
\begin{figure}[H]
%\ContinuedFloat
\centering
\begin{tabular}{| l l | }
\hline
True value  &  {\color{black} \rule[0.08cm]{0.05\textwidth}{1pt}}   \\ 
Bias-corrected estimator &  {\color{mygreen} \rule[0.08cm]{0.05\textwidth}{1pt}} \\ 
Empirical estimator estimator &  {\color{myred} \rule[0.08cm]{0.05\textwidth}{1pt}} \\ \hline
\end{tabular}
\caption{\renyi estimates for order $2$ for support $10000$, 
number of samples ranging from $1000$ to $10000$, averaged over $100$ trials.}
\label{fig:estimation_Halpha_2}
\end{figure}

\begin{figure}[H]
\centering     %%% not \center
\subfloat[Uniform]{\label{fig:uniform}\includegraphics[width=45mm]{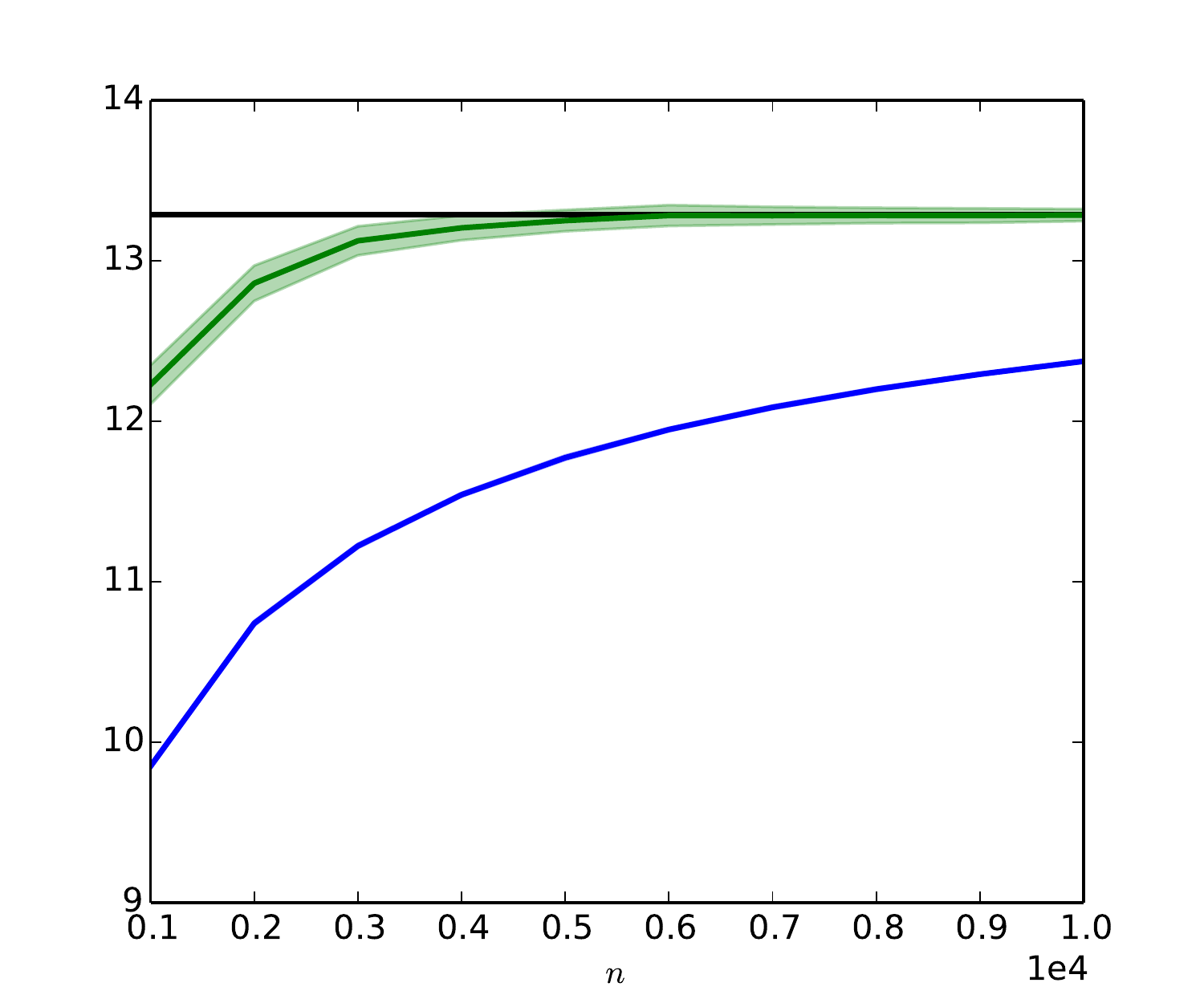}}
\subfloat[Step]{\label{fig:step}\includegraphics[width=45mm]{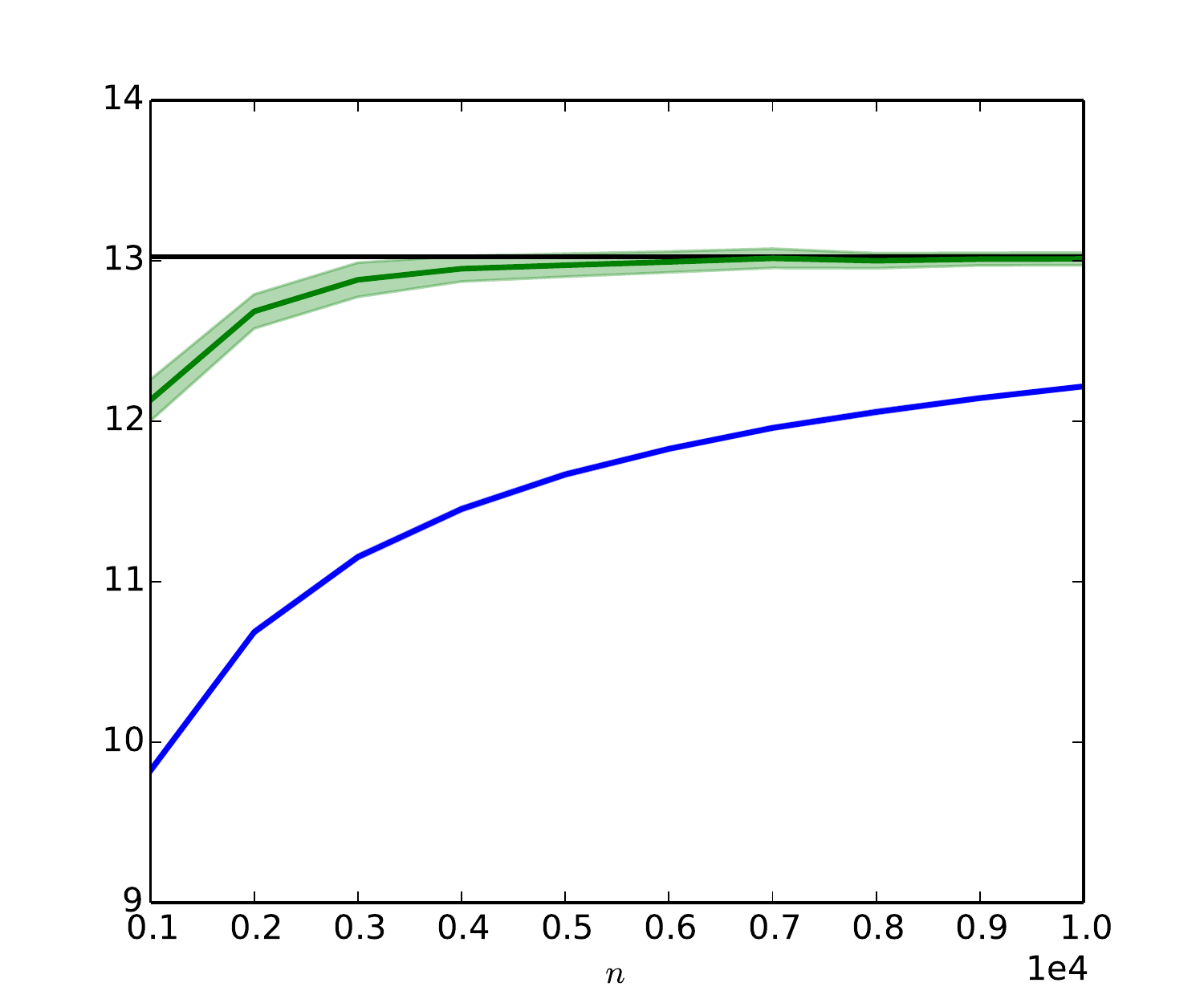}}
\subfloat[Zipf with parameter $3/4$]{\label{fig:zipf1}\includegraphics[width=45mm]{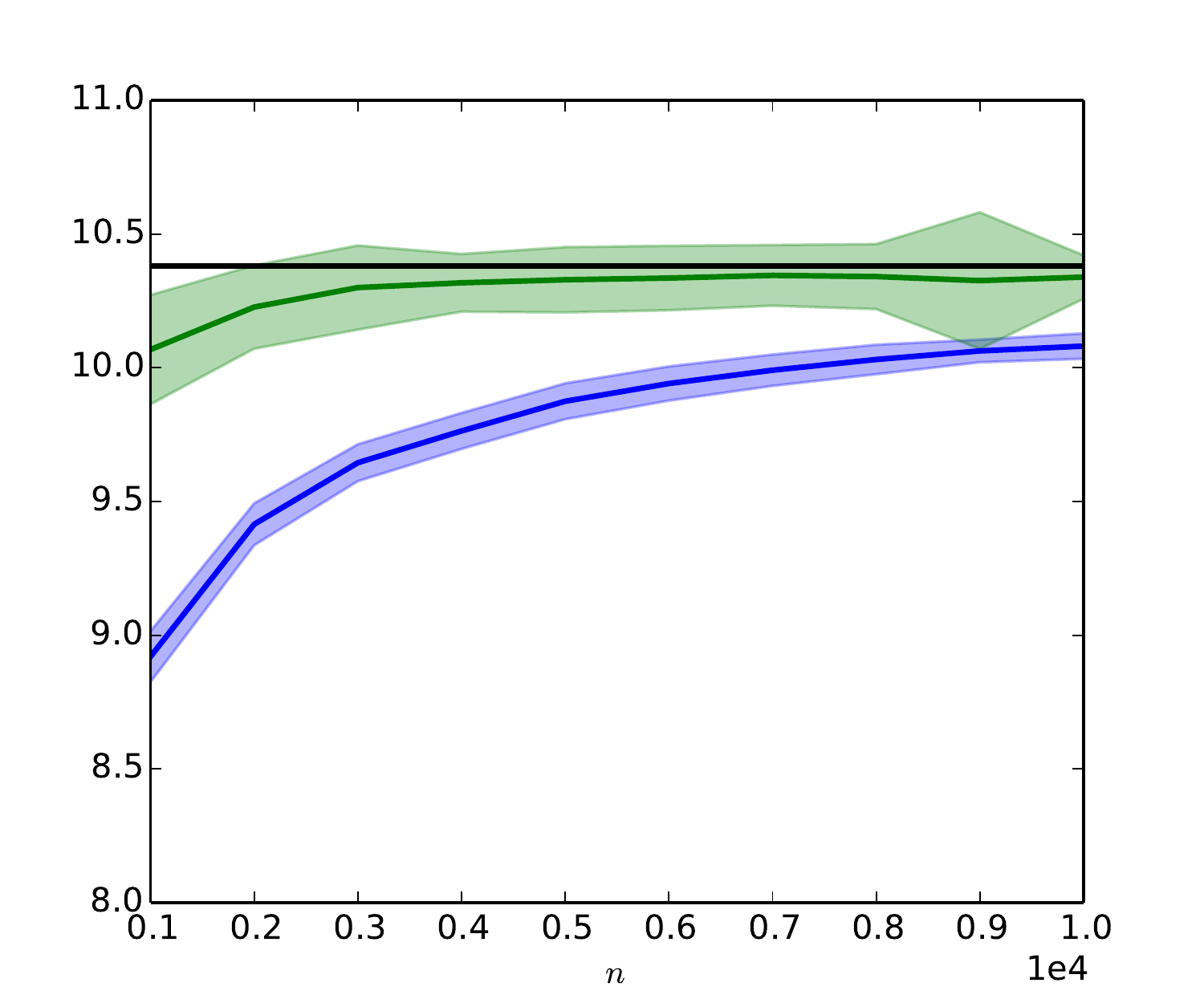}}
\end{figure}
\begin{figure}[H]
%\ContinuedFloat
\centering
\subfloat[Zipf with parameter $1/2$]{\label{fig:zipf15}\includegraphics[width=45mm]{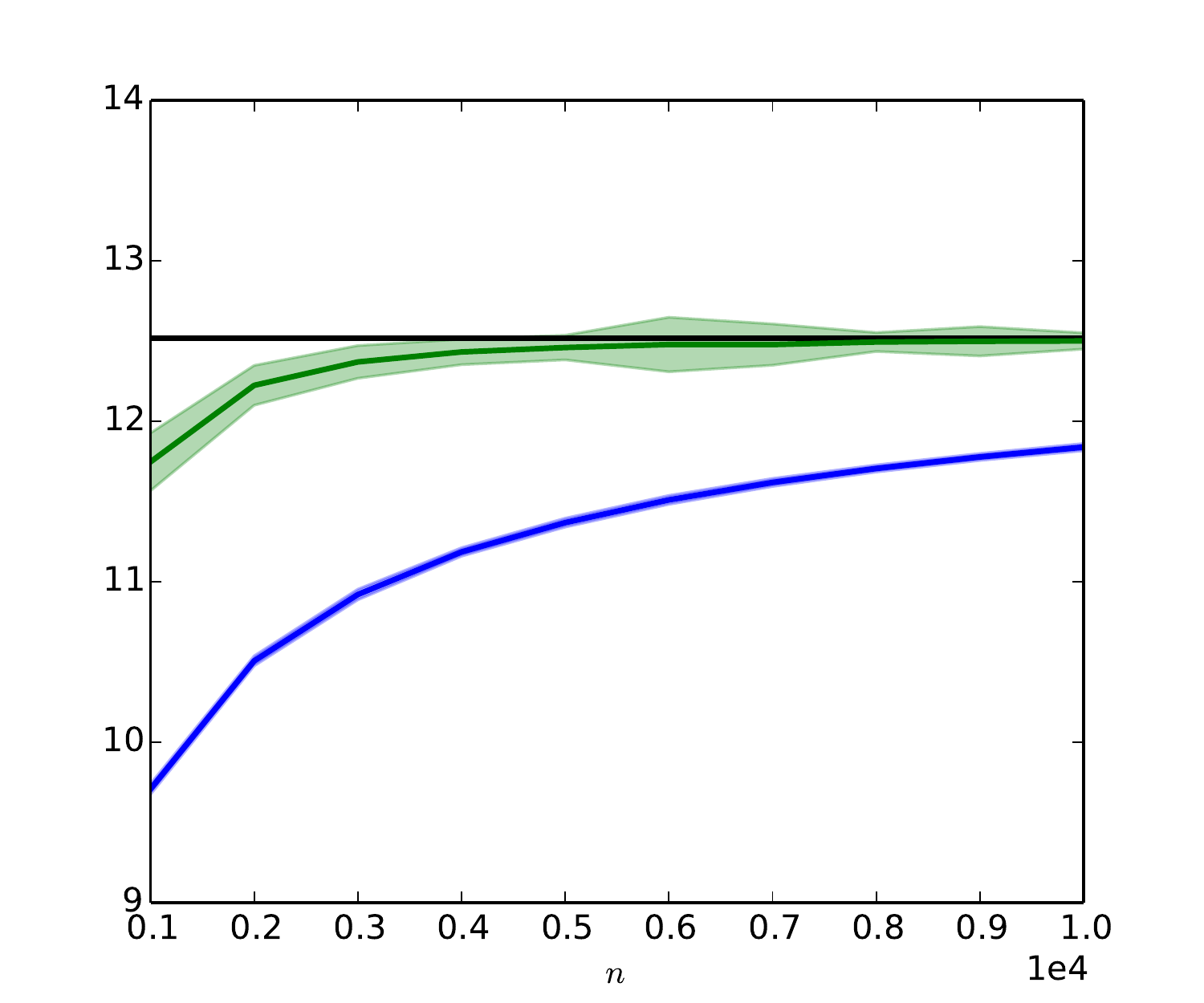}}
\subfloat[Uniform prior (Dirichlet $1$)]{\label{fig:random}\includegraphics[width=45mm]{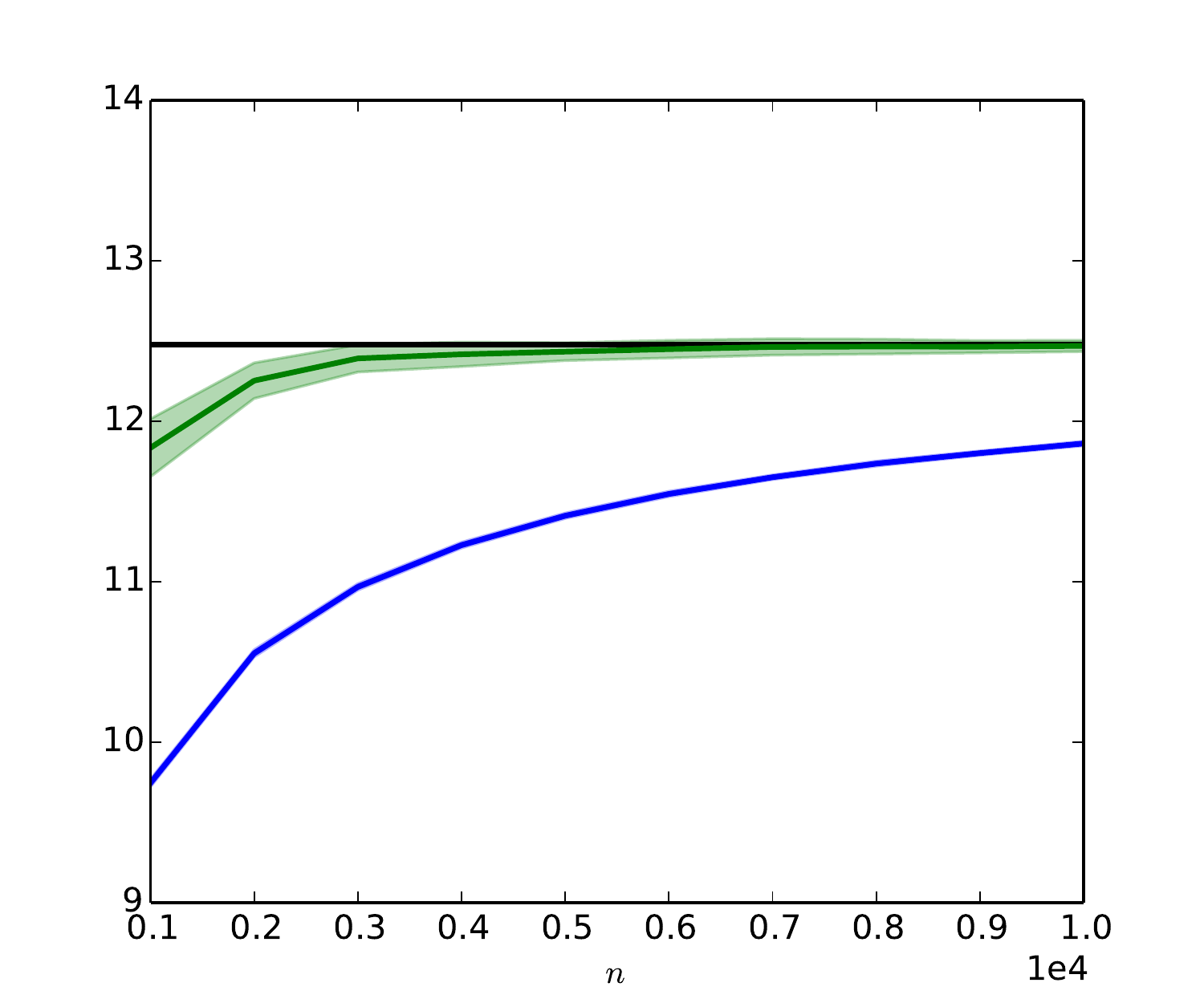}}
\subfloat[Dirichlet $1/2$ prior]{\label{fig:gammahalf}\includegraphics[width=45mm]{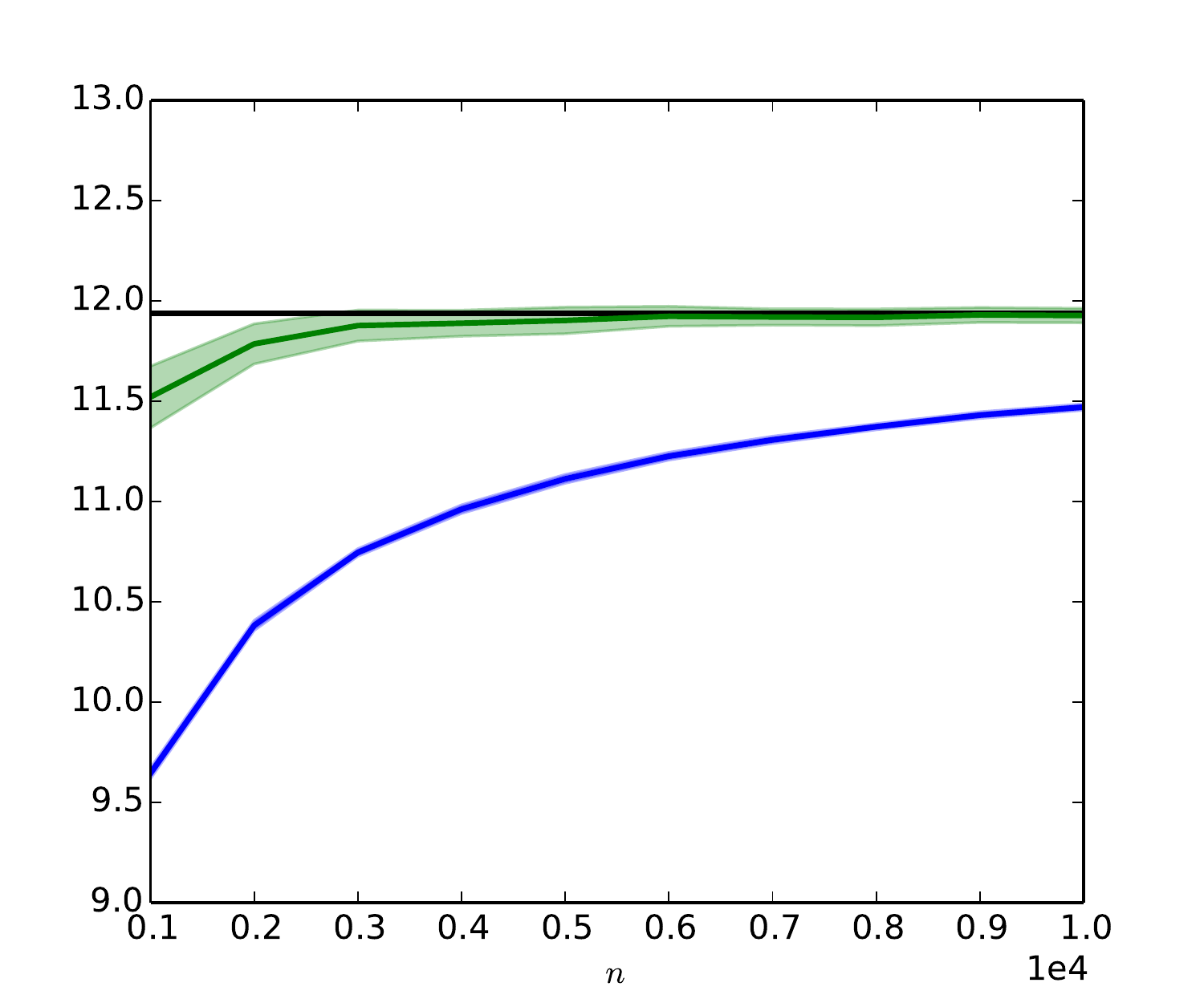}}
\end{figure}
\begin{figure}[H]
%\ContinuedFloat
\centering
\begin{tabular}{| l l | }
\hline
True value  &  {\color{black} \rule[0.08cm]{0.05\textwidth}{1pt}}   \\ 
Polynomial-approximation estimator &  {\color{mygreen} \rule[0.08cm]{0.05\textwidth}{1pt}} \\ 
Empirical estimator estimator &  {\color{myred} \rule[0.08cm]{0.05\textwidth}{1pt}} \\ \hline
\end{tabular}
\caption{\renyi estimates for order $1.5$ for support $10000$, 
number of samples ranging from $1000$ to $10000$, averaged over $100$ trials.}
\label{fig:estimation_Halpha_15}
\end{figure}

 % We prove a unified converse for both integral and non-integral $\renprm$
 
\section{Lower bounds on sample complexity}
\label{s:lower_bounds} 
We now establish lower bounds on $\sakde$.
The proof is based on exhibiting two distributions $\dP$ and $\dQ$
with $\rental \dP \neq \rental \dQ$ such that the set of $\Mltsmb$'s
have very similar distribution from $\dP$ and $\dQ$, 
if fewer samples than the claimed lower bound are available. This
method is often referred to as {\it Le Cam's two-point method} (see,
for instance, \cite{Yu97}). The key idea is summarized in the
following result which is easy to derive.  

\begin{lemma}\label{l:two_point} If for two distributions $\dP$ and $\dQ$ on $\cX$ and
  $\nsmp\in\mN$ the total variation distance $\|\dP^n -\dQ^n\| < \prerr$, then one of the following holds for every function $\hfnc$:
\begin{align}
&\dP\left(|\rental{\dP} - \hfnc(X^n)| \geq \frac{|\rental{\dP} - \rental{\dQ}|}{2}\right) \geq \frac{1-\prerr}{2},
\nonumber
\\\text{or }\,\,&\dQ\left(|\rental{\dQ}- \hfnc(X^n)| \geq \frac{|\rental{\dP} - \rental{\dQ}|}{2}\right) \geq \frac{1-\prerr}{2}.
\nonumber
\end{align}
\end{lemma}
We first prove the lower bound for integers $\renprm>1$, which matches the upper bound in Theorem~\ref{t:upper_bounds_integer} up
to a constant factor. 
\begin{theorem}
\label{t:lower_bounds_integer2}
Given an $1<\renprm\in \mN$ and $1< \prerr<1$, for every sufficienly
small $\esterr>0$
\[
\sam_\renprm(k, \esterr, \prerr) =
\Omega\left(\frac{\absz^{(\renprm-1)/\renprm}}{\esterr^2}\right),
\]
where the constant implied by $\Omega$ may depend on $\ep$.
\end{theorem}
\begin{proof}
We rely on Lemma~\ref{l:two_point} and exhibit two distributions $\dP$
and $\dQ$ with appropriate properties. Specifically, consider the following distributions $\dP$ and $\dQ$ over
$[\absz]$: 
$\dP_1 = 1/{\absz^{1-1/\renprm}}$, and for $x=2,\ldots,\absz$, $\dPx =
(1-\dP_1)/(\absz-1)$; $\dQ_1 = (1+\esterr)/{\absz^{1-1/\renprm}}$, and for $x=2,\ldots,\absz$, $\dQ_x =
(1-\dQ_1)/(\absz-1)$. Then, we have 
\begin{align*}
\normP\renprm &= \frac1{k^{\renprm-1}} +
(k-1)\cdot\left(\frac{1-\frac1{k^{1-1/\renprm}}}{k-1}\right)^{\renprm}
\\
&= \frac1{k^{\renprm-1}}
+\frac1{(k-1)^{\renprm-1}}\cdot\left(1-\frac1{k^{1-1/\renprm}}\right)^{\renprm} 
\\
&= \frac{(2+o_k(1))}{k^{\renprm-1}}.
\end{align*}
Similarly, 
\begin{align*}
\normQ\renprm &= \frac{(1+\esterr)^\renprm}{k^{\renprm-1}} +
(k-1)\cdot\left(\frac{1-\frac{(1+\esterr)}{k^{1-1/\renprm}}}{k-1}\right)^{\renprm}
\\
%% = \frac{(1+\esterr)^\renprm}{k^{\renprm-1}}
%% +\frac1{k^{\renprm-1}}\cdot\left(1-\frac1{k^{1-1/\renprm}}\right)^{\renprm} 
=\frac{(2+\renprm\esterr+o_k(1))}{k^{\renprm-1}}.
\end{align*}
Therefore, $|\rental\dP-\rental\dQ|=\Omega(\esterr)$. To complete the
proof, we show that  there exists a constant $C_\ep$ such that
$\|\dP^n- \dQ^n\|\leq \ep$ if $n\leq C_\ep \absz^{1-1/\renprm}/\esterr^2$.
To that end, we bound the squared Hellinger distance between
$\dP^n$ and $\dQ^n$ given by
\[
h^2(\dP, \dQ) = 2- 2\sum_{x}\sqrt{\dPx\dQx} = \sum_x (\sqrt{\dPx} -\sqrt{\dQx})^2.
\] 
Since for small values of $\esterr$ we have $(1+\esterr)^{1/2}<1+\esterr$,
\begin{align*}
h^2(\dP,\dQ) &= \left(\sqrt{\frac{1+\esterr}{\absz^{1-1/\renprm}}} -
\sqrt{\frac{1}{\absz^{1-1/\renprm}}} \right)^2
+ \left(\sqrt{1- \frac{1+\esterr}{\absz^{1-1/\renprm}}} -
\sqrt{1- \frac{1}{\absz^{1-1/\renprm}}} \right)^2
\\
&=O\left(\frac{\esterr^2}{\absz^{1-1/\renprm}}\right).
\end{align*}
The required bound for $\|\dP^n-\dQ^n\|$ follows using the following
standard steps ($cf.$~\cite{Yu97})
\begin{align*}
\|\dP^n-\dQ^n\| &\leq \sqrt{h^2(\dP,\dQ)} 
\\
&= \sqrt{1- \left(1-\frac 12h^2(\dP, \dQ)\right)^n}
\\
&\leq \sqrt{\frac n2 h^2(\dP, \dQ)}.
\end{align*}
\end{proof}

Next, we lower bound $\sam_\renprm(\absz)$ for noninteger $\renprm>1$ and show that it must be almost linear in
$\absz$. While we still rely on Lemma~\ref{l:two_point} for our lower
bound, we take recourse to Poisson sampling to simplify our calculations.
\begin{lemma}{\bf (Poisson approximation 2)} 
Suppose there exist $\esterr, \prerr >0$ such that, with $\Nsmp \sim \poid{2\nsmp}$, for all
estimators $\hfnc$ we have
\begin{align}
\max_{\dP \in \cP}\bPr{|\rental{\dP} - \hfnc_\renprm(\XoN)| > \esterr } > \prerr,
\nonumber
\end{align}
where $\cP$ is a fixed family of distributions. Then, for all fixed
 length estimators $\tfnc$
\begin{align*}
\max_{\dP \in \cP}\bPr{|\rental{\dP} - \tfnc_\renprm(\Xon)| > \esterr } >\frac{\prerr}{2},
\end{align*}
when $\nsmp > 4\log(2/\prerr)$.
\end{lemma}

Also, it will be convenient to replace the observations $\XoN$ with its 
{\it profile}  $\Pfl = \Pfl(\XoN)$ 
\cite{OrlitskySVZ04}, i.e., $\Pfl = (\Pfl_1, \Pfl_2,\ldots)$
where $\Pfll$ is the number of elements $x$ that appear $l$ times in
the sequence $\XoN$. The following well-known result
says that for estimating $\rental{\dP}$, it suffices to consider only the functions of the profile.
\begin{lemma}{\bf (Sufficiency of
    profiles).}\label{l:sufficiency_of_profiles} Consider an estimator  
$\hfnc$ such that 
\begin{align}
\bPr{|\rental{\dP} - \hfnc(\XoN)| > \esterr} \leq \prerr, \quad \text{for all } \dP.
\nonumber
\end{align}
Then, there exists an estimator $\tfnc(\XoN) = \tfnc(\Pfl)$ such that
\begin{align}
\bPr{|\rental{\dP} - \tfnc(\Pfl)| > \esterr} \leq \prerr, \quad \text{for all } \dP.
\nonumber
\end{align}
\end{lemma}
\noindent Thus, lower bounds on the sample complexity will follow upon 
showing a contradiction for the second inequality above when the number
of samples $\nsmp$ is sufficiently small. We obtain the required
contradiction by using Lemma~\ref{l:two_point} upon showing there are distributions $\dP$ and $\dQ$ of support-size $\absz$ such that the following hold:
\begin{itemize}
\item[(i)] There exists $\esterr >0$ such that 
\begin{align}
|\rental{\dP} - \rental{\dQ}| > \esterr;
\label{e:estimation_error}
\end{align}
\item[(ii)] denoting by $\dP_\Pfl$ and $\dQ_\Pfl$, respectively, the distributions on the profiles
under Poisson sampling corresponding to underlying distributions $\dP$ and $\dQ$, there exist 
$\prerr>0$ 
such that
\begin{align}
\|\dP_\Pfl - \dQ_\Pfl\| < \prerr,
\label{e:L1_profiles}
\end{align}
if $\nsmp <  k^{\,c(\renprm)}$.
\end{itemize}
Therefore, it suffices to find two distributions $\dP$ and $\dQ$ with
different R\'enyi entropies and with small total variation distance between
the distributions of their profiles, when $\nsmp$ is sufficiently small.
For the latter requirement, we recall a result of~\cite{Valiant08} that
allows us to bound the total variation distance in \eqref{e:L1_profiles}
in terms of the differences of power sums $|\norma(\dP)- \norma(\dQ)|$.
\begin{theorem}{\cite{Valiant08}}\label{t:profiles_var_dis}
Given distributions $\dP$ and $\dQ$ such that
\begin{align}
\max_x \max\{\dPx ; \dQ_x\} \leq \frac{\prerr}{40\nsmp},
\nonumber
\end{align}
for Poisson sampling with $\Nsmp \sim \psns{\nsmp}$, it holds that
$$\|\dP_\Pfl - \dQ_\Pfl\|  \leq \frac{\prerr}{2} + 5\sum_a \nsmp^a |\norma(\dP)- \norma(\dQ)|.$$
\end{theorem}
It remains to construct the required distributions $\dP$
and $\dQ$, satisfying \eqref{e:estimation_error} and \eqref{e:L1_profiles} above. 
By Theorem~\ref{t:profiles_var_dis}, the total variation distance $\|\dP_\Pfl - \dQ_\Pfl\|$ can be
made small by ensuring that the power sums of distributions $\dP$ and $\dQ$ are matched, that is,
we need distributions $\dP$ and $\dQ$ with different R\'enyi entropies and identical
power sums for as large an order as possible. To that end, for every positive integer ${\dm}$ and every vector $\bx = (x_1, ..., x_{\dm}) \in \mR^{\dm}$, associate
with $\bx$ a distribution $\dP^{\bx}$ of support-size ${\dm}\absz$ such that
$$\dP^{\bx}_{ij} = \frac{|x_i|}{k\|\bx\|_1}, \quad 1\leq i \leq {\dm}, \, 1\leq j\leq k.$$
Note that  
\begin{align}
\rental{\dP^{\bx}} &=  \log k + \frac{\renprm}{\renprm-1}\log\frac{\|\bx\|_1}{\|\bx\|_\renprm},
\nonumber
\end{align}
and for all $a$
\begin{align}
\norma\left(\dP^{\bx}\right) &=  \frac{1}{k^{a-1}}\left(\frac{\|\bx\|_a}{\|\bx\|_1}\right)^a.
\nonumber
\end{align}
We choose the required distributions $\dP$ and $\dQ$, respectively, as $\dP^{\bx}$ and $\dP^{\by}$,
where the vectors $\bx$ and $\by$ are given by the next result.
\begin{lemma}\label{l:xy_construction}
For every ${\dm}\in \mN$ and $\renprm$ not integer, 
there exist positive vectors $\bx, \by \in \mR^{\dm}$ such that
\begin{align}
\|\bx\|_r &= \|\by\|_r, \quad 1\leq r \leq {{\dm}-1},
\nonumber
\\
\|\bx\|_{\dm} &\neq \|\by\|_{\dm},
\nonumber
\\
\|\bx\|_\renprm &\neq \|\by\|_\renprm.
\nonumber
\end{align}
\end{lemma} 
\begin{proof} 
Let $\bx = (1, ..., {\dm}))$. Consider the polynomial
\begin{align*}
p(z) &= (z-x_1)...(z-x_{\dm}),
\end{align*}
and $q(z) = p(z) - \Delta$, where $\Delta$ is chosen small enough so that $q(z)$ has ${\dm}$
positive roots. Let $y_1, ..., y_{\dm}$ be the roots of the polynomial $q(z)$. By Newton-Girard identities, while the sum of 
${\dm}$th power of roots of a polynomial does depend on the constant term, the sum of
first ${\dm}-1$ powers of roots of a polynomial do not depend on it. Since $p(z)$ and $q(z)$ differ only by a constant, it holds that
\begin{align*}
\sum_{i=1}^{{\dm}}x_i^r = \sum_{i=1}^{{\dm}}y_i^r, \quad 1 \leq r \leq {\dm}-1,
\end{align*}
and that
\begin{align*}
\sum_{i=1}^{{\dm}}x_i^{{\dm}} \neq \sum_{i=1}^{{\dm}}y_i^{\dm}.
\end{align*}
Furthermore, using a first order Taylor approximation, we have
\begin{align*}
y_i - x_i = \frac{\Delta}{p'(x_i)} + \order(\Delta),
\end{align*}
and for any differentiable function $g$,
\begin{align*}
g(y_i) - g(x_i) = g'(x_i)(y_i - x_i) + o(|y_i - x_i|).
\end{align*}
It follows that
\begin{align*}
\sum_{i=1}^{\dm} g(y_i) - g(x_i) = \sum_{i=1}^{\dm}\frac{g'(x_i)}{p'(x_i)}\Delta + \order(\Delta),
\end{align*}
and so, the left side above is nonzero for all $\Delta$ sufficiently small provided
\begin{align}
\sum_{i=1}^{\dm}\frac{g'(x_i)}{p'(x_i)} \neq 0.
\nonumber
\end{align}
Upon choosing $g(x) = x^\renprm$, we get
\begin{align*}
\sum_{i=1}^{\dm}\frac{g'(x_i)}{p'(x_i)} = \frac{\renprm}{{\dm} !}\sum_{i=1}^{\dm} \left(
\begin{matrix}
{\dm}\\ i
\end{matrix}
\right)(-1)^{{\dm}-i}\, i^\renprm.
\end{align*}
Denoting the right side above by $h(\renprm)$, note that $h(i) = 0$
for $i = 1, ..., {\dm}-1$. Since $h(\renprm)$ is a linear combination of 
${\dm}$ exponentials, it cannot have more than ${\dm}-1$ zeros (see, for instance, \cite{Tos06}).
Therefore, $h(\renprm) \neq 0$ for all $\renprm \notin \{1,..., {\dm}-1\}$; in particular,  $\|\bx\|_\renprm \neq \|\by\|_\renprm$
for all $\Delta$ sufficiently small.
\end{proof}
We are now in a position to prove our converse results.

% We first prove the lower bound for an integer $\renprm >1$.
% \begin{theorem}\label{t:lower_bounds_integer}
% Given an integer $\renprm>1$ and any estimator $\Fnc$ of $\rental \dP$, 
% for every $0 < \prerr< 1$ there exits a distribution $\dP$
% with support of size $\absz$, $\esterr >0$ and a constant $\Const>0$ such that for 
% $n < \Const \absz^{(\renprm-1)/\renprm}$
% we have 
% \[
% \bPr{|\rental{\dP} - \Fnc\left(\Xon\right)| \geq \esterr } \geq \frac{1-\prerr}{2}.
% \]
% In particular, for every $0 < \prerr< 1/2$ there exists $\esterr >0$ such that
% $$\sam_\renprm(k, \esterr, \prerr) = \Omega\left(\absz^{(\renprm-1)/\renprm}\right).$$
% \end{theorem}
% \begin{proof} For ${\dm} = \renprm$,
% let $\dP$ and $\dQ$, respectively, be the distributions $\dP^{\bx}$ and $\dP^{\by}$,
% where the vectors $\bx$ and $\by$ are given by Lemma \ref{l:xy_construction}. 
% In view of the foregoing discussion, we need to verify \eqref{e:estimation_error} 
% and \eqref{e:L1_profiles} to prove the theorem. Therefore, \eqref{e:estimation_error} holds by  Lemma \ref{l:xy_construction} since
% \[
% |\rental{\dP} - \rental{\dQ}| = \frac{\renprm}{1-\renprm}\left|\log \frac{\|\bx\|_\renprm}{\|\by\|_\renprm}\right| >0,
% \]
%  and for $\nsmp < \Const_2\absz^{({\dm}-1)/{\dm}}$ and $5\Const_2^{\dm}/(1-\Const_2) < \prerr/2$. inequality \eqref{e:L1_profiles} follows from Theorem \ref{t:profiles_var_dis}
% as
% \[
% \|\dP_\Pfl - \dQ_\Pfl\|  \leq \frac{\prerr}{2} + 5\sum_{a\geq {\dm}} \left(\frac{\nsmp}{\absz^{(a-1)/a}}\right)^a
% \leq \prerr.\qedhere
% \]
% \end{proof}
 
\begin{theorem}\label{t:lower_bounds_arbitrary}
Given a nonintegral $\renprm>1$, for any fixed $0 < \prerr< 1/2$, we have
$$\sam_\renprm(\absz, \esterr, \prerr) = \dbltldOmg(\absz).$$
\end{theorem} 
% \begin{remark*}
% We can obtain a slightly stronger lower bound on $\sam_\renprm(k, \esterr, \prerr)$ for noninteger $\renprm>1$, with a minor modification of our proof, upon using \cite[Corollary 1]{Valiant08}. In particular, the statement of the theorem above holds with
% $$\sam_\renprm(k, \esterr, \prerr) \geq \frac{\Const\absz}{ 2^{\const\sqrt{\log k}}},$$ 
% for appropriately chosen constants $\Const$ and $\const$. 
% \end{remark*}
\begin{proof} For a fixed $\dm$, let distributions
$\dP$ and $\dQ$ be as in the previous proof. Then, as in the proof of Theorem~\ref{t:lower_bounds_arbitrary},  inequality \eqref{e:estimation_error} holds by Lemma~\ref{l:xy_construction} and~\eqref{e:L1_profiles} holds by Theorem \ref{t:profiles_var_dis}
if $\nsmp < \Const_2\absz^{({\dm}-1)/{\dm}}$. The theorem follows since
$\dm$ can be arbitrary large.
\end{proof}

Finally, we show that $\sam_\renprm(\absz)$ must be super-linear in $\absz$ for $\renprm<1$.
\begin{theorem}\label{t:lower_bounds_alpha_small}	
Given $\renprm<1$, for every $0 < \prerr< 1/2$, we have
$$\sam_\renprm(k, \esterr, \prerr) = \dbltldOmg\left(\absz^{1/\renprm}\right).$$
\end{theorem}
\begin{proof} Consider distributions $\dP$ and $\dQ$ on an alphabet of size $kd +1$, where
\[
\dP_{ij} = \frac{\dP^{\bx}_{ij}}{\absz^\beta} \text{ and } \dQ_{ij} = \frac{\dP^{\bx}_{ij}}{\absz^\beta}, \quad 1\leq i \leq {\dm}, \, 1\leq j\leq k,
\]
where the vectors $\bx$ and $\by$ are given by Lemma \ref{l:xy_construction} and $\beta$ satisfies $\renprm(1+\beta) < 1$, and
\[
\dP_0 =  \dQ_0 = 1 - \frac{1}{\absz^\beta}.
\]
For this choice of $\dP$ and $\dQ$, we have
\begin{align*}
\norma\left(\dP\right) &=  \left(1 - \frac{1}{\absz^\beta}\right)^a + \frac{1}{k^{a(1+\beta) -1}}\left(\frac{\|\bx\|_a}{\|\bx\|_1}\right)^a,
\\
\rental{\dP} &=  \frac{1-\renprm(1+\beta) }{1 - \renprm}\log \absz + \frac{\renprm}{1-\renprm}\log\frac{\|\bx\|_\renprm}{\|\bx\|_1} + \Order(\absz^{a(1+\beta) -1}),
\end{align*}
and similarly for $\dQ$, which further yields
\begin{align*}
|\rental{\dP} - \rental{\dQ}|&=  \frac{\renprm}{1-\renprm}\left|\log\frac{\|\bx\|_\renprm}{\|\by\|_\renprm}\right| + \Order(\absz^{a(1+\beta) -1}).
\end{align*}
Therefore, for  sufficiently large $\absz$, \eqref{e:estimation_error} holds by Lemma~\ref{l:xy_construction} since $\renprm(1+\beta) < 1$, and for $\nsmp < \Const_2 \absz^{(1 + \beta -1/\dm)}$ we get \eqref{e:L1_profiles} by Theorem \ref{t:profiles_var_dis} as
\begin{align}
\|\dP_\Pfl - \dQ_\Pfl\|  \leq \frac{\prerr}{2} + 5\sum_{a\geq {\dm}} \left(\frac{\nsmp}{\absz^{1+\beta -1/a}}\right)^a
\leq \prerr.
\nonumber
\end{align} 
The theorem follows since $\dm$ and $\beta< 1/\renprm -1$ are arbitrary.
\end{proof}

\section*{Acknowledgements}
The authors thank Chinmay Hegde and Piotr Indyk for helpful discussions and suggestions.

%%%%%%%%%%%%%%%%%%%%%%%%%%%%%%%%%%%%%%%%%
\bibliography{IEEEabrv,masterref}
\bibliographystyle{IEEEtranS}
%%%%%%%%%%%%%%%%%%%%%%%%%%%%%%%%%%%%%%%
\section*{Appendix A: Estimating power sums}
The broader problem of estimating smooth functionals of
 distributions was considered in~\cite{Valiant11b}.
 Independently and concurrently with this work,~\cite{JiaoVW14}
 considered estimating more general 
 functionals and applied their technique to estimating 
 the power sums of a distribution to a given additive accuracy.
 Letting $\saak$ denote the number of samples needed to estimate
 $\normP\renprm$ to a given additive accuracy,~\cite{JiaoVW14} showed that for
 $\renprm<1$,
 \begin{equation}
 \label{eqn:mmn_lto_upr}
 \saak = \Theta\Paren{\frac{\absz^{1/\renprm}}{\log \absz}},
 \end{equation}
 and~\cite{JiaoVW14ii} showed that for $1<\renprm<2$,
 \begin{equation}
 \label{eqn:mmn_upr_lto}
 \saak
 \le
 \Order\Paren{\absz^{2/\renprm-1}}.
 \nonumber
 \end{equation}
 In fact, using techniques similar to multiplicative guarantees on $\normP\renprm$ we show that 
 for $\saak$ is a constant 
 independent of $\absz$ for all $\absz>1$. 

 Since $\normP\renprm>1$ for $\renprm<1$, power sum estimation to a fixed
 additive accuracy implies also a fixed multiplicative
 accuracy, and therefore
 \[
 \sak
 =
 \Theta(\smak)
 \le
 \Order(\saak),
 \]
 namely for estimation to an additive accuracy, \renyi 
 requires fewer samples than power sums.
 Similarly, $\normP\renprm<1$ for $\renprm>1$, and therefore
 \[
 \sak
 =
 \Theta(\smak)
 \ge
 \Omega(\saak),
 \]
 namely for an additive accuracy in this range,
 \renyi requires more samples than power sums.

 It follows that the power sum estimation results in~\cite{JiaoVW14,JiaoVW14ii}
 and the R\'enyi-entropy estimation results in this paper
 complement each other in several ways.
 For example, for $\renprm<1$, 
 \[
 \dbltldOmg\Paren{\absz^{1/\renprm}}
 \le
 \sak
 =
 \Theta(\smak)
 \le
 \Order(\saak)
 \le
 \Order\Paren{\frac{\absz^{1/\renprm}}{\log \absz}},
 \]
 where the first inequality follows from Theorem~\ref{t:lower_bounds_alpha_small} and the last follows
 from the upper-bound~\eqref{eqn:mmn_lto_upr} derived in~\cite{JiaoVW14} using a {\it polynomial approximation estimator}.
 Hence, for $\renprm<1$, estimating power sums to additive and
 multiplicative accuracy require a comparable number of samples.

 On the other hand, for $\renprm>1$, 
 Theorems~\ref{t:upper_bounds_arbitrary}
 and ~\ref{t:lower_bounds_arbitrary}
 imply that for non integer $\renprm$, 
 $
 \dbltldOmg\Paren{\absz}
 \le
 \smak
 \le
 \Order\Paren{\absz},
 $
 while in the Appendix we show that
 for $1<\renprm$, 
 $\saak$ is a constant.
 Hence in this range, power sum estimation to a multiplicative accuracy
 requires considerably more samples than estimation to an additive accuracy.

%\purge{As in Section~\ref{sct:rlt_mmn_est}, }
%\tcolor{Let $\saak$ denote the number of samples needed to estimate}
%the power sum $\normP\renprm$ to a given additive accuracy. 
We now 
show that
the empirical estimator requires a constant number of samples to estimate
$\normP\renprm$ independent of $\absz$, $i.e.$, 
$\saak = \Order(1)$. In view of Lemma~\ref{l:bias-variance-PAC},
it suffices to bound the bias and variance of the empirical estimator. Concurrently with this work,
 similar results were obtained in an updated version of \cite{JiaoVW14}.

As before, we comsider Poisson sampling with $N \sim \poid {n}$ samples. The \emph{empirical} or \emph{plug-in} estimator of $\normP\renprm$
is
\begin{align*}
%\label{eqn:estimator}
\estmmnemp
\ed
\sum_\smb \left(\frac{\Mltsmb}{\nsmp}\right)^\renprm.
\end{align*}
The next result shows that the bias and the variance of the empirical estimator are $\order(1)$.
\begin{lemma}
\label{lem:bias}
For an appropriately chosen constant $c>0$, the bias and the variance of the empirical 
estimator are bounded above as
\begin{align*}
\left| \estmmnemp -\normP \renprm\right|&\le 2c\max\{n^{-(\renprm -1)}, n^{-1/2}\},\\
\Var[ {\mmntest}]&\le 2c\max\{n^{-(2\renprm -1)}, n^{-1/2}\},\\
\end{align*}
for all $n\geq 1$.
\end{lemma}
\begin{proof}
Denoting $\npsmb = n\dPx$,
we get
the following bound on the bias for an appropriately chosen constant $c$:
\begin{align*}
\left|\estmmnemp -\normP \renprm\right|
&\le \frac{1}{n^\renprm}\sum_{\npsmb\le 1} \left|\expectation{\Mltx^\renprm} -\npsmb \right|
+\frac{1}{n^\renprm}\sum_{\npsmb> 1} \left|\expectation{\Mltx^\renprm} -\npsmb \right|\\
&\le \frac{c}{n^\renprm}\sum_{\npsmb\le 1} \npsmb
+ \frac{c}{n^\renprm}\sum_{\npsmb>1}\Paren{\npsmb+ \npsmb^{\renprm -1/2}},
\end{align*}
where the last inequality holds by Lemma~\ref{l:bound_Poisson_moments2} and Lemma~\ref{lem:fall_fac}%~\eqref{e:mmnt_bnd_Poi_lamle1}
since $x^\renprm$ is convex in $x$. Noting $\sum_i \npsmb =\nsmp$,
we get
\[
\left|\estmmnemp -\normP \renprm\right| \le 
\frac{c}{n^{\renprm-1}} +  \frac{c}{n^\renprm}\sum_{\npsmb>1}\npsmb^{\renprm -1/2}.
\]
Similarly, proceeding as in the proof of Theorem~\ref{t:upper_bounds_arbitrary}, 
the variance of the empirical estimator is bounded as
\begin{align*}
\Var[ {\mmntest}]
&=
\frac{1}{\nsmp^{2\renprm}}
\sum_{x\in\cX}  \expectation{ \Mltx^{2\renprm}} -\EE [\Mltx^\renprm]^2 \\
&\le
\frac{1}{\nsmp^{2\renprm}}
\sum_{x\in\cX}
\left|\expectation{\Mltx^{2\renprm}} - \npsmb^{2\renprm}\right|\\
&\le 
\frac{c}{n^{2\renprm-1}} +  \frac{c}{n^{2\renprm}}\sum_{\npsmb>1}\npsmb^{2\renprm -1/2}.
\end{align*}
The proof is completed upon showing that
\[
\sum_{\npsmb>1}\npsmb^{\renprm -1/2} 
\le \max\{\nsmp, \nsmp^{\renprm -1/2}\}, \quad  \renprm>1.
\]
To that end, note that for $\renprm<3/2$
\[
\sum_{\npsmb>1}\npsmb^{\renprm -1/2} \le \sum_{\npsmb>1}\npsmb \le \nsmp, \quad \renprm < 3/2.
\]
Further, since $x^{\renprm -1/2}$ is convex for $\alpha\ge 3/2$,
the summation above is maximized when one of the
$\npsmb$'s is $\nsmp$ and the remaining equal $0$ which yields
\[
\sum_{\npsmb>1}\npsmb^{\renprm - 1/2}\le n^{\renprm-1/2}, \quad \alpha \ge 3/2,
\]
and completes the proof.
\end{proof}

\section*{Appendix B: Lower bound for sample complexity of empirical estimator}
We now derive lower bounds for the sample complexity of the empirical estimator of $\rental \dP$.

\begin{lemma}
Given $\renprm<1$ and $\esterr < c_\renprm$ for a constant $c_\renprm$ depending only on $\renprm$, the sample complexity 
$S_{\renprm}^{\estrenemp}(\absz,\esterr, \prerr)$
of the empirical estimator $\estrenemp$  is bounded below as
\[
S_{\renprm}^{\estrenemp}(\absz,\esterr, 0.9) = \Omega\left(\frac \absz {
  \esterr}\right).
\]
\end{lemma}
\begin{proof}
We prove the lower bound for the uniform distributon over $\absz$
symbols in two steps. We first show that for any constant $c_1 > 1$ if 
$\nsmp < \absz/c_1$ then 
the additive approximation error is at least $\delta$ with probability
one, for every $\delta < \log c_1$. Then,
assuming that $\nsmp \geq \absz/c_1$, we show that the additve
approximation error is at least $\delta$ with probability greater than
$0.9$ if $\nsmp < k/\delta$.

For the first claim, we assume without loss of generality that $\nsmp
\leq \absz$, since otherwise the proof is complete. Note that for $\renprm>1$ the function $(\dP_i
- y)^\renprm + (\dP_j +y)^\renprm$ is decreasing in $y$ for all $y$
such that $(\dP_i- y)  > (\dP_j +y)$. Thus, the minimum value of 
$\sum_\smb \left(\frac{\Mltsmb}{\nsmp}\right)^\renprm$ is attained
when each $\Mltsmb$ is either $0$ or $1$. It follows that
\[
\estmmnemp
=
\sum_\smb \left(\frac{\Mltsmb}{\nsmp}\right)^\renprm 
\geq \frac{1}{\nsmp^{\renprm-1}},
\]
which is the same as
\[
\rental {\dP} - \frac{1}{\renprm-1} \log \frac{1}{\estmmnemp} 
\geq  \log \frac{\absz}{\nsmp}.
\]
Hence, for any $c_1>1$ and $\nsmp < \absz/c_1$ and any $0\leq \delta \leq \log c_1$, the
additive approximation error is more than $\delta$ with probability one.  

Moving to the second claim, suppose now $\nsmp > \absz/c_1$. We first show that with high probability, the
multiplicities of a linear fraction of $\absz$ symbols should be at
least a factor of standard deviaton
higher than the mean.  Specifically, let
\[
A = \sum_{\smb} \indicator \left(\Mltsmb \geq \frac{\nsmp}{\absz} +
c_2\sqrt{\frac{\nsmp}{\absz}\left(1-\frac{1}{\absz} \right)} \right).
\]
Then,
\begin{align*}
\expectation{A} &= \sum_{\smb}\expectation{\indicator\left(\Mltsmb \geq \frac{\nsmp}{\absz} +
c_2\sqrt{\frac{\nsmp}{\absz} \left( 1-\frac{1}{\absz} \right)} \right)}
\\
&= k \cdot p\left(\Mltsmb \geq  \frac{\nsmp}{\absz} +
c_2\sqrt{\frac{\nsmp}{\absz} \left( 1-\frac{1}{\absz} \right)} \right) 
\\
&\geq  k Q(c_2),
\end{align*}
where $Q$ denotes the $Q$-function, $i.e.$, the tail of the standard
normal random variable, and the final inequality uses Slud's inequality \cite[Theorem 2.1]{Slud77}. 

Note that $A$ is a function of $\nsmp$ i.i.d. random
variables $X_1,X_2,\ldots,X_\nsmp$, and changing any one $X_i$ changes $A$ by at most $2$. Hence, by McDiarmid's
inequality,  
\[
\Pr(A \geq \expectation{A} - \sqrt{8 \nsmp} ) \geq 1- e^{-4} \geq 0.9.
\]
Therefore, for all $k$ sufficiently large (depending on $\delta$) and
denoting $c= Q(c_2)/2$, at least $c\absz$ symbols occur more than   
$\frac{\nsmp}{\absz} + c_2\sqrt{\frac{\nsmp}{\absz}}$ 
times with probability greater than $0.9$.
Using the fact that $(\dP_i- y)^\renprm + (\dP_j+y)^\renprm$ is
decreasing if $(\dP_i- y)  > (\dP_j +y)$ once more, we get 
\begin{align*} 
\sum_{\smb\in\cX} \frac{\Mltsmb^\renprm}{\nsmp^\renprm}
& = \sum_{\smb: \Mltsmb \geq t} \frac{\Mltsmb^\renprm}{\nsmp^\renprm} + \sum_{\smb: \Mltsmb < t} \frac{\Mltsmb^\renprm}{\nsmp^\renprm}\\
&\geq c\absz  \left(\frac{1}{\absz} +
c_2\sqrt{\frac{1}{\nsmp\absz} } \right)^\renprm + (1-c)\absz \left(\frac{1}{\absz} -
\frac{cc_2}{1-c}\sqrt{\frac{1}{\nsmp\absz}} \right)^\renprm 
\\
&=\frac{1}{\absz^{\renprm-1}} \left[ c\left(1 +
c_2\sqrt{\frac{\absz}{\nsmp}} \right)^\renprm + (1-c)\left(1 -
\frac{cc_2}{1-c}\sqrt{\frac{\absz}{\nsmp}} \right)^\renprm\right]
\\
&\geq 
\frac{1}{\absz^{\renprm-1}} \left[c\left(1 +
c_2\sqrt{\frac{\absz}{\nsmp}} \right)^\renprm + (1-c)\left(1 -
\frac{\renprm cc_2}{1-c}\sqrt{\frac{\absz}{\nsmp}} \right)
\right]
\\
&\geq \frac{1}{\absz^{\renprm-1}} \left[c\left(1 +
\renprm c_2\sqrt{\frac{\absz}{\nsmp} } + c_4 \frac{\absz}{\nsmp}\right) + (1-c)\left(1 -
\frac{\renprm cc_2}{1-c}\sqrt{\frac{\absz}{\nsmp}} \right)\right]
\\
&= \frac{1}{\absz^{\renprm-1}} \left(1+ cc_4\frac{\absz}{\nsmp}\right)
\end{align*}
where the second inequality is by Bernoulli's inequality and the third
inequality holds for every $c_4\leq \renprm(\renprm-1)(c_2\sqrt{c_1})^{\renprm-2}/2$.
Therefore, with probability $\geq 0.9$,
\[
\rental \dP - \frac{1}{\renprm-1} \log \frac{1}{\estmmnemp} \geq \frac{1}{\renprm-1} \log  \left(1 + cc_4\frac{k}{n} \right),
\]
which yields the desired bound.
\end{proof}

\begin{lemma}
Given $\renprm<1$ and $\esterr < c_\renprm$ for a constant $c_\renprm$ depending only on $\renprm$, the sample complexity 
$S_{\renprm}^{\estrenemp}(\absz,\esterr, \prerr)$
of the empirical estimator $\estrenemp$  is bounded as
\[
S_{\renprm}^{\estrenemp}(\absz,\esterr, 0.9) =
\Omega\left(\frac{\absz^{1/\renprm}}{\esterr^{1/\renprm}}\right).
\]
\end{lemma}
\begin{proof}
We proceed as in the proof of the previous lemma. However, instead of
using the uniform distribution, we use a distribution which has one
``heavy element'' and is uniform conditioned on the occurance of the
remainder. The key observation is that there will be roughly
$\nsmp^\renprm$ occurances of the ``light elements''. Thus, when we
account for the error in the estimation of the contribution of light
elements to the power sum, we can replace $\nsmp$ with
$\nsmp^{1/\renprm}$ in our analysis of the previous lemma, which
yields the required bound for sample complexity.

Specifically, consider a distribution with one heavy element $0$ such that 
\[
\dP_0 = 1- \frac{\esterr}{\nsmp^{1-\renprm}},\,\, \text{ and }\,\,
\dP_i= \frac{\delta}{\absz \nsmp^{1-\renprm}}, \quad 1\leq i \leq
\absz.
\]
Thus,
\begin{align}
\normP \renprm = \left(1- \frac{\esterr}{\nsmp^{1-\renprm}}\right)^\renprm + \esterr^\renprm\left(\frac{\absz}{\nsmp^{\renprm}}\right)^{1-\renprm}.
\label{e:actual_power_sum}
\end{align}
We begin by analyzing the estimate of the second term in power
sum, namely
\[
\sum_{i \in [\absz]} \left(\frac{\Nsmp_i}{n}\right)^\renprm.
\]
Let $R = \sum_{i \in [\absz]} \Nsmp_i$ be the total number of
occurances of light elements. Since $R$ is a binomial $(n,
\delta\nsmp^{\renprm-1})$ random variable, for every constant $c>0$
\[
\bPr{ 1-c<\frac{R}{\delta n^\renprm} <1+c} \geq 1 - \frac{1}{c^2n}.
\]
In the remainder of the proof, we
shall assume that this large probability event holds.

As in the proof of the previous lemma, we first prove a $\esterr$
independent lower bound for sample complexity. 
To that end, we fix $\esterr =1$ in the definition of $\dP$. 
Assuming
$(1+c)\nsmp^\renprm \leq \absz$, which implies $R\leq \absz$,
and
using the fact that $(\dP_i- y)^\renprm - (\dP_j
+y)^\renprm$ is increasing in $y$ if $(\dP_i- y)  > (\dP_j +y)$, we get
\begin{align*}
\estmmnemp &\leq 1 + \left(\frac{R}{n}\right)^\renprm \sum_{i\in [k]} \left(\frac{\Nsmp_i}{R}\right)^\renprm
\\
&\leq 1 +
\frac{(1+c)^\renprm}{\nsmp^{\renprm(1-\renprm)}}
\sum_{i\in [k]} \left(\frac{\Nsmp_i}{R}\right)^\renprm 
\\
&\leq 1 + \frac{(1+c)^\renprm}{\nsmp^{\renprm(1-\renprm)}}
R^{1-\renprm}
\\
&\leq
3,
\end{align*}
where the last inequality uses $R\leq (1+c)\nsmp^\renprm \leq 2\nsmp^\renprm$.
Thus, the empirical estimate is at most $3$ with probability close to
$1$ when $\absz$ (and therefore $\nsmp$) large. It follows from
\eqref{e:actual_power_sum}  that
\begin{align*}
\rental \dP - \frac{1}{1-\renprm} \log {\estmmnemp} 
&\geq  \log \frac{k}{3n^\renprm} .
\end{align*}
Therefore, for all $c_1>1$, $\esterr < \log 3c_1$ and $\absz$
sufficiently large, 
at least $\left(\absz/c_1\right)^{1/\renprm}$ samples are needed to 
get a $\esterr$-additive approximation of $\rental \dP$ with
probability of error less than $1-1/({c^2n})$. Note that we only
needed to assume $R\leq (10/9)\nsmp^\renprm$, an event with probability
greater than $0.9$, to get the contradiction above. Thus, we may
assume that $\nsmp \geq \left(\absz/c_1\right)^{1/\renprm}$. Under
this assumption, for $\absz$ sufficiently large, $\nsmp$ is
sufficiently large so that $(1-c)\nsmp^\renprm \leq R\leq
(1+c)\nsmp^\renprm$ holds with probability arbitrarily close to $1$.

Next, assuming that $\nsmp \geq \left(\absz/c_1\right)^{1/\renprm}$, 
we obtain a $\esterr$-dependent lower bound for sample complexity of the empirical estimator. 
We use the $\dP$ mentioned above with a general $\esterr$ and assume that the large probability event 
\begin{align}
(1-c) \leq \frac R{\esterr\nsmp^\renprm} \leq (1+c)
\label{e:good_event}
\end{align}
holds. Note that conditioned on each value of $R$,
the random variables $(\Nsmp_i, i\in [k])$ have a multinomial
distribution with uniform probabilities, $i.e.$, these random
variables behave as if we drew $R$ i.i.d. samples from a uniform
distribution on $[k]$ elements. Thus, we can follow the proof of the
previous lemma {\it mutatis mutandis}. We now define $A$ as 
\[
A = \sum_{\smb} \indicator \left(\Mltsmb \leq \frac{\nsmp}{\absz} -
c_2\sqrt{\frac{\nsmp}{\absz}\left(1-\frac{1}{\absz} \right)} \right).
\]
and satisfies 
Then,
\begin{align*}
\expectation{A} &= \sum_{\smb}\expectation{\indicator\left(\Mltsmb
  \leq \frac{\nsmp}{\absz} - c_2\sqrt{\frac{\nsmp}{\absz} \left( 1-\frac{1}{\absz} \right)} \right)}
\\
&= k \cdot p\left(\Mltsmb \leq  \frac{\nsmp}{\absz} -
c_2\sqrt{\frac{\nsmp}{\absz} \left( 1-\frac{1}{\absz} \right)} \right).
\end{align*}
To lower bound $p\left(\Mltsmb \leq  \frac{\nsmp}{\absz} -
c_2\sqrt{\frac{\nsmp}{\absz} \left( 1-\frac{1}{\absz} \right)}
\right)$ Slud's inequality is no longer available (since it may not hold
for ${\tt Bin}(n, p)$ with $p> 1/2$ and that is the regime of interest for the
lower tail probability bounds needed here). Instead we take recourse to a
combination of Bohman's inequality and Anderson-Samuel inequality, as
suggested in \cite[Eqns. (i) and (ii)]{Slud77}. It can be verified
that the condition for \cite[Eqns. (ii)]{Slud77} holds, and therefore,
\[
p\left(\Mltsmb \leq  \frac{\nsmp}{\absz} -
c_2\sqrt{\frac{\nsmp}{\absz} \left( 1-\frac{1}{\absz} \right)} \right)
\geq Q(c_2).
\]
Continuing as in the proof of the previous lemma, we get that the following holds with 
conditional probability greater than $0.9$ given each value of $R$ satisfying \eqref{e:good_event}:
\begin{align*}
\sum_{i\in [k]} \left(\frac{\Mltsmb}{R}\right)^\renprm
&\leq \absz^{1-\renprm}\left( 1 - c_3 \frac k R\right)
\\
&\leq 
\absz^{1-\renprm}\left( 1 - c_4 \frac k {\esterr\nsmp^\renprm}\right),
\end{align*}
where $c_3$ is a sufficiently small constant such that $(1+x)^\renprm
\leq 1+\renprm x  - c_3x^2$ for all $x\geq 0$ and $c_4 = c_3/(1+c)$. Thus, 
\begin{align*}
\estmmnemp &\leq 1 + \left(\frac{R}{\nsmp}\right)^\renprm
\sum_{i\in [k]} \left(\frac{\Nsmp_i}{R}\right)^\renprm
\\
&\leq 1 + \left(\frac{R}{\nsmp}\right)^\renprm
\absz^{1-\renprm}\left( 1 - c_4 \frac k {\esterr\nsmp^\renprm}\right)
\\
&\leq 1 + \left(1+c\right)^\renprm \esterr^\renprm
\left(\frac{\absz}{\nsmp^\renprm}\right)^{1-\renprm}\left( 1 - c_4 \frac k {\esterr\nsmp^\renprm}\right).
\end{align*}
Denoting $y = (\absz/\nsmp^\renprm)$ and choosing $c_1$ and $c$ small
enough such that $\estmmnemp \leq 2$, for all sufficiently large
$\nsmp$ we get from \eqref{e:actual_power_sum} that 
\begin{align*}
\frac{\normP \renprm}{\estmmnemp} &\geq \frac{1 - \delta +
  y^{1-\renprm}}{1+ \esterr^{\renprm}(1+c)^\renprm y^{1-\renprm} - (1+c)^\renprm c_4 \esterr^{\renprm-1}y^{2-\renprm}}
\\
&\geq \frac{1 - \delta +
  y^{1-\renprm}}{1+ y^{1-\renprm} - \esterr^{\renprm-1}y^{2-\renprm}}
\\
&\geq 1- \frac{\delta}2 + \frac{\esterr^{\renprm -1}y^{2-\renprm}}2,
\end{align*}
where the second inequality uses the fact that
$\esterr^{\renprm}(1+c)^\renprm y^{1-\renprm} - (1+c)^\renprm c_4
\esterr^{\renprm-1}y^{2-\renprm}$ is negative, $c_4>1$ and
$\esterr<1$. Therefore, $\frac{\normP \renprm}{\estmmnemp}\geq
1+\esterr$ if $y^{2-\renprm}\geq 3\esterr^{2-\renprm}$, which
completes the proof.  
\end{proof}

\end{document}